\documentclass[sigconf]{acmart}
\usepackage{amsmath}
\usepackage{cleveref}

\usepackage{bm}
\usepackage{bbm}
\usepackage{booktabs,subcaption,amsfonts,dcolumn}
\usepackage{wasysym}
\usepackage{hhline}
\usepackage{float}
\usepackage{commath}

\PassOptionsToPackage{dvipsnames}{xcolor}
\usepackage{xcolor}
\usepackage{color}
\usepackage[dvipsnames]{xcolor}
\usepackage{enumitem}
\newtheorem{theorem}{Theorem}[section]

 \usepackage{multirow}
 \usepackage[ruled,vlined,linesnumbered]{algorithm2e}
\usepackage{booktabs,siunitx}

\usepackage{stmaryrd}
\usepackage{trimclip}
\usepackage{pifont}
\usepackage{titlesec}
\titlespacing*{\section}{0pt}{2pt}{2pt}
\titlespacing*{\subsection}{0pt}{2pt}{2pt}
\titlespacing*{\subsubsection}{0pt}{1pt}{1pt}
\usepackage{scalerel,stackengine}
\stackMath
\newcommand\reallywidehat[1]{%
\savestack{\tmpbox}{\stretchto{%
  \scaleto{%
    \scalerel*[\widthof{\ensuremath{#1}}]{\kern-.6pt\bigwedge\kern-.6pt}%
    {\rule[-\textheight/2]{1ex}{\textheight}}
  }{\textheight}%
}{0.5ex}}%
\stackon[1pt]{#1}{\tmpbox}%
}
\DeclareMathOperator*{\argmax}{arg\,max}

\usepackage{titlesec}
\titlespacing*{\section}{0pt}{5pt plus 1pt}{5pt minus 1pt}
\titlespacing*{\subsection}{0pt}{5pt plus 1pt}{5pt minus 1pt}
\titlespacing*{\subsubsection}{0pt}{2pt plus 1pt}{2pt minus 1pt}
\AtBeginDocument{%
  \providecommand\BibTeX{{%
    \normalfont B\kern-0.5em{\scshape i\kern-0.25em b}\kern-0.8em\TeX}}}

\copyrightyear{2023}
\acmYear{2023}
\setcopyright{rightsretained}
\acmConference[CIKM '23]{Proceedings of the 32nd ACM International
Conference on Information and Knowledge Management}{October 21--25,
2023}{Birmingham, United Kingdom}
\acmBooktitle{Proceedings of the 32nd ACM International Conference on
Information and Knowledge Management (CIKM '23), October 21--25, 2023, Birmingham, United Kingdom}\acmDOI{10.1145/3583780.3614877}
\acmISBN{979-8-4007-0124-5/23/10}

\begin{document}

\title{FARA: Future-aware Ranking Algorithm for Fairness Optimization}

\author{Tao Yang}

\affiliation{%
   \institution{University of Utah}
     \streetaddress{50 Central Campus Dr.}
   \city{Salt Lake City}
   \state{Utah}
   \postcode{84112}
}
\email{taoyang@cs.utah.edu}
\author{Zhichao Xu}

\affiliation{%
   \institution{University of Utah}
     \streetaddress{50 Central Campus Dr.}
   \city{Salt Lake City}
   \state{Utah}
   \postcode{84112}
}
\email{zhichao.xu@utah.edu}

\author{Zhenduo Wang}
\affiliation{%
	\institution{University of Utah}
	\streetaddress{50 Central Campus Dr.}
	\city{Salt Lake City}
	\state{Utah}
	\country{USA}
	\postcode{84112}
}
\email{zhenduow@cs.utah.edu}

\author{Qingyao Ai}
\affiliation{%
   \institution{DCST,  Tsinghua University \\ Quan Cheng Laboratory, Zhongguancun Laboratory}
   \city{Beijing}
   \state{China}
   \postcode{100084}
}
\authornote{Corresponding authors}
\email{aiqy@tsinghua.edu.cn}

\begin{abstract}
Ranking systems are the key components of modern Information Retrieval (IR) applications, such as search engines and recommender systems. Besides the ranking relevance to users, the exposure fairness to item providers has also been considered an important factor in ranking optimization. 
Many fair ranking algorithms have been proposed to jointly optimize both ranking relevance and fairness.
However, we find that most existing fair ranking methods adopt greedy algorithms that only optimize rankings for the next immediate session or request. As shown in this paper, such a myopic paradigm could limit the upper bound of ranking optimization and lead to suboptimal performance in the long term. 

To this end, we propose \textbf{FARA}, a novel \textbf{F}uture-\textbf{A}ware \textbf{R}anking \textbf{A}lgorithm for ranking relevance and fairness optimization. Instead of greedily optimizing rankings for the next immediate session, FARA plans ahead by jointly optimizing multiple ranklists together and saving them for future sessions.
Specifically, FARA first uses the Taylor expansion to investigate how future ranklists will influence the overall fairness of the system. 
Then, based on the analysis of the Taylor expansion, FARA adopts a two-phase optimization algorithm where we first solve an optimal future exposure planning problem and then construct the optimal ranklists according to the optimal future exposure planning. 
Theoretically, we show that FARA is optimal for ranking relevance and fairness joint optimization. 
Empirically, our extensive experiments on three semi-synthesized datasets show that FARA is efficient, effective, and can deliver significantly better ranking performance compared to state-of-the-art fair ranking methods.
We make our implementation public at \href{https://github.com/Taosheng-ty/QP_fairness/}{https://github.com/Taosheng-ty/QP\_fairness/}.

\end{abstract}

\keywords{Fair Ranking, Position Bias, Exposure, Exposure Fairness}

\ccsdesc[500]{Information systems~Learning to rank}


\maketitle

\section{INTRODUCTION}
Ranking systems are one of the important cornerstones of information retrieval (IR). 
Existing ranking systems are usually constructed to optimize ranking relevance with the Probability Ranking Principle (PRP) \cite{robertson1977prp} where items of greater likely relevance should be ranked higher.  
The PRP is a user-centered ranking strategy that helps save users energy and time since users could satisfy their needs with the top-ranked items~\cite{joachims2017accurately}. 
However, recent research has shown that, besides users, item providers also draw utility from ranking systems, and the PRP could result in severe unfairness for item providers~\cite{singh2018fairness,biega2018equity}. Particularly, 
the PRP always assigns a few top items with high-rank positions. Those top items usually get the majority of exposure while other items rarely get exposure, although other items might still be  relevant~\cite{patro2022fair,kotary2022end,yang2023mitigating,tran2021ultra}. The unbalanced exposure leads to unfair opportunities and unfair economic gains for item providers. Such unfairness will eventually force unfairly treated providers to leave the system, and fewer options will be left for users~\cite{yang2022effective}. Therefore, IR researchers have argued that ranking relevance and fairness are both important for modern ranking systems~\cite{singh2018fairness,singh2019policy}. Many fair ranking algorithms have been proposed to optimize both of them jointly~\cite{patro2022fair,zehlike2021fairness}.


However, existing fair algorithms are mostly greedy algorithms and could only deliver suboptimal ranking performance in the long run. 
In particular, existing fair ranking algorithms~\cite{singh2018fairness,biega2018equity,morik2020controlling,oosterhuis2021computationally,wu2021tfrom} usually behave greedily to sequentially produce the locally optimal ranklist for the next immediate session without being aware of the influence of future sessions\footnote{In this paper, we define a session as a query issued by a user.} (see more discussion in \S \ref{sec:RelatedWork}). The unawareness could lead to unmitigated ranking conflict between relevance and fairness optimization. For example, imagine a case that there are in total 3 items in consideration, item \textit{A}, item \textit{B}, and item \textit{C}, where item \textit{A} is the most relevant one and item \textit{C} is the least relevant one.  Ranklist $[A,B,C]$ is the ranklist to maximize ranking relevance. We now consider a scenario where item \textit{C} is severely unfairly treated in history. To optimize exposure fairness, we need to  allocate item \textit{C} more exposure by boosting item \textit{C} to a higher position. However, If we try to greedily boost item \textit{C} within the next immediate session, it is highly likely that item \textit{C} will be boosted to the first rank to get the maximum exposure and the result ranklist is $[C,A,B]$. However, Ranklist $[C,A,B]$ is of poor ranking relevance due to the ranking conflict that the least relevant item (item \textit{C}) is put on the most important rank (the first rank).

Intuitively, the ranking conflict can be smoothed if we plan ahead and jointly optimize multiple future sessions' ranklists together instead of greedily optimizing the next immediate session.
For example, the multiple ranklists after joint optimization can be $[[A,C,B],[A,C,B],...]$, where item \textit{C} is smoothly boosted in multiple ranklists and the most relevant item, i.e., item \textit{A},  is still ranked the highest. 
Based on the above idea, we propose \textbf{FARA}, a novel \textbf{F}uture-\textbf{A}ware \textbf{R}anking \textbf{A}lgorithm for relevance and fairness optimization. 
Briefly, FARA precomputes and jointly optimizes multiple ranklists together and saves them for future use. 
Particularly, to be able to plan for the future, FARA first uses the Taylor expansion to investigate how future ranklists will influence fairness. 
Then, based on the influence, FARA uses a two-phase optimization to jointly optimizes multiple ranklists together for future use. 
In phase 1, we solve an exposure planning problem and get the optimal future planning for item exposure. 
In phase 2, we construct the optimal ranklists according to the optimal future planning for item exposure. 
We prove FARA's optimum in terms of ranking relevance and fairness joint optimization in \S~\ref{sec:theoreticalAnalysis}.
Extensive experiments on three semi-synthesized datasets also demonstrate FARA's effectiveness and efficiency compared to existing fair ranking algorithms (\S~\ref{sec:experimental_results}). 



\vspace{-5pt}
\section{RELATED WORK}
\label{sec:RelatedWork}
\vspace{-2pt}
\noindent\textbf{\textit{Ranking Fairness:}} \,
Due to the importance of rankings for providers (sellers, job seekers, content creators, etc.), ~\cite{zehlike2017fa,singh2018fairness,joachims2021fairness,zehlike2020reducing,ekstrand2023overview}, there has been growing interest in ranking fairness for providers \cite{usunier2022fast,ge2022explainable,heuss2022fairness,raj2022measuring,naghiaei2022cpfair,usunier2022fast,li2021towards,bigdeli2022gender,wu2021tfrom}.
However, the definitions of ranking fairness vary a lot in the existing literature, and there exists no universal definition. 
At a high level, existing fairness definitions can be grouped into  probability-based fairness and exposure-based fairness~\cite{patro2022fair,zehlike2021fairness}. 
Probability-based fairness~\cite{asudeh2019designing,celis2017ranking,geyik2019fairness} usually requires a minimum number or proportion of protected (e.g., race, gender) items to be distributed evenly across a ranklist. 
However, only considering the number or proportion of items in a ranklist neglects the fact that different ranks usually have different importance. 
To address this, exposure-based fairness~\cite{singh2018fairness,biega2018equity,yang2021maximizing,gao2021addressing,diaz2020evaluating} assigns values to each ranking position based on the expected user attention or click probability. 
Exposure-based fairness argues that total exposure is a limited resource for a ranking system and  advocates for fair distribution of exposure among items to ensure fairness for item providers~\cite{patro2022fair}. 
In this paper, we limit our discussion of fair ranking algorithms within the scope of exposure-based fairness.

\vspace{3pt}
\noindent\textbf{\textit{Fair Ranking Algorithms:}} \space
Recently, a few ranking algorithms ~\cite{patro2022fair,zehlike2021fairness,usunier2022fast,heuss2022fairness,mansoury2022understanding,saito2022fair,gao2022fair} have been proposed to achieve exposure-based fairness.
In this work, we classify them as  \textbf{open-loop} algorithms or \textbf{feedback-loop} algorithms depending on whether historically generated ranklists are used to correct ranking scores. 
For open-loop fair algorithms~\cite{singh2019policy,singh2018fairness,singh2021fairness,oosterhuis2021computationally,heuss2022fairness,vardasbi2022probabilistic,wu2022joint}, each item usually has a static and fixed ranking score once the ranking model is optimized.
Then ranklists are stochastically sampled for each session according to the static ranking scores. 
Various techniques have been used to optimize the static ranking model, such as linear programming~\cite{singh2018fairness,heuss2022fairness}, policy gradient~\cite{singh2019policy}, differentiable PL model optimization~\cite{oosterhuis2021computationally}. 
However, given the fact that ranking scores are static, open-loop algorithms are usually not robust. 
To improve ranking robustness, feedback-loop algorithms~\cite{yang2021maximizing,morik2020controlling,biega2018equity,yang2023marginal}
dynamically take historical ranklists as input to correct items' scores. For example, \citet{morik2020controlling} proposes to use a proportional controller to boost ranking scores of historically unfairly treated items. 



\begin{table}[t]
	\caption{A summary of notations in this paper.}
	\small
	\def\arraystretch{1}
	\begin{tabular}
		{| p{0.07\textwidth} | p{0.36\textwidth}|} \hline
		$d,q,D(q)$ & For a query $q$, $D(q)$ is the set of candidates items. $d\in D(q)$ is an item. \\\hline
		$e,r,c$ & All are binary random variables indicating whether an item $d$ is examined ($e=1$), perceived as relevant ($r=1$) and clicked ($c=1$) by a user respectively. \\\hline
		$R,P_i,E,\pi$ & $R=P(r=1|d)$,  is the probability of an item $d$ perceived as relevant. $P_i=P(e=1|rank(d|\pi)=i)$ is the examination probability of item $d$ when it is put in $i^{th}$ rank in a ranklist $\pi$. $E$ is item's accumulated examination probability (see Eq.\ref{eq:expo}).  \\\hline
		$k_s,k_c$ & Users will stop examining items lower than rank $k_s$ due to selection bias (see Eq.~\ref{eq:selectionBias}).  $k_c$ is the cutoff prefix to evaluate Cum-NDCG and $k_c\le k_s$.\\\hline
	\end{tabular}\label{tab:notation}
\vspace{-2pt}
\end{table}
\section{BACKGROUND AND PRIOR KNOWLEDGE}
\vspace{-2pt}
In this section, we provide readers with background knowledge of the paper. 
A summary of notations we use throughout the paper can be found in Table \ref{tab:notation}. 

\vspace{3pt}
\noindent \textbf{\textit{Ranking Services Workflow:}} \,
We take web search as an example to detail the ranking service workflow. 
At time step $t$, a ranking session starts when user $t$ issues a query $q$. For query $q$, there exist candidate items provided by item providers.
With the query and candidate items, the ranking system first estimates each item's relevance and then constructs a ranklist of candidate items by optimizing certain ranking objectives.
Then the ranking system presents the ranklist to users and collects users' feedback (e.g., clicks) which can be utilized to update the relevance estimator.
    

\vspace{3pt}
\noindent\textbf{\textit{Partial and Biased Feedback:}} \,
Relevance estimation is usually updated using users' feedback.
However, such feedback is usually a noisy and biased indicator of relevance since users only provide meaningful feedback for items they have examined. 
If we consider user clicks as the main signal for user feedback, then we have
\begin{equation}
\vspace{-2pt}
    c= \begin{cases}
      r, &\text{if}\quad e=1 \\
      0, &\text{otherwise}
    \end{cases}
\vspace{-2pt}
\label{eq:biasedBinary}
\end{equation}
where $e,r,c$ are binary random variables indicating whether an item is examined, perceived as relevant, and clicked, respectively.
With User Examination Hypothesis \cite{richardson2007predicting}, we can model users' click probability as
\begin{equation}
\vspace{-2pt}
    P(c=1)=P(e=1)P(r=1)
    \label{eq:clickProb}
    \vspace{-2pt}
\end{equation}
For the rest of the paper, we use $R=P(r=1)$ to simplify the notation.
Although there exist several types of biases in the examination probability $P(e=1)$, we focus on two most important ones: positional bias and selection bias.

\textbf{Positional Bias}~\cite{craswell2008experimental}: Examination probability is decided by the rank (also called position) and drops along ranks. Particularly, the examining probability is denoted as $P_{\textit{rank}(d|\pi)}$, where $\textit{rank}(d|\pi)$ is item $d$'s rank in ranklist $\pi$.

\textbf{Selection Bias} \cite{oosterhuis2021unifying,oosterhuis2020policy}: 
This bias exists when not all of the items are selected to be shown to users, or some lists are so long that no users will examine every item in them.
Assuming the items ranked lower than rank $k_s$ will not be examined \cite{oosterhuis2020policy}, we model this as:
\begin{equation}
    P(e=1|d,\pi)=\begin{cases}
      P_{\textit{rank}(d|\pi)}, & \text{if}\quad \textit{rank}(d|\pi)\le k_s \\
      0, & \text{otherwise}
    \end{cases}
    \label{eq:selectionBias}
\end{equation}

\noindent\textbf{\textit{Ranking Utility Measurement:}} \space
Here we introduce the evaluation of ranking performance from both the user side and the provider side, which in later sections will guide the ranking optimization.

\textbf{The User-side Utility}: 
User-side utility measures a ranking system's ability to put relevant items on higher ranks. 
A popular user-side utility measurement is \textit{DCG} \cite{jarvelin2002cumulated}. 
Specifically, given a query $q$ and a ranklist $\pi$, $\textit{DCG}@k_c$ is defined as,
\begin{equation}
    \textit{DCG}@k_c(\pi)=\sum_{i=1}^{k_c}R(\pi[i],q)\lambda_i=\sum_{i=1}^{k_c}R(\pi[i],q)P_i
    \label{eq:DCG}
\end{equation}
where $\pi[i]$ indicates the $i^{th}$ ranked item in ranklist $\pi$, $R(\pi[i],q)$ indicates $\pi[i]$'s relevance to query $q$, cutoff $k_c$  indicates the prefix we want to evaluate, $\lambda_i$ indicates the weight we put on $i^{th}$ rank. $\lambda_i$ is usually monotonically decreases as $i$ increases, e.g.,  $\lambda_i$ is usually set to $\frac{1}{\log_2 (i+1)}$~\cite{jarvelin2002cumulated}.
In this paper, following previous works in \cite{singh2018fairness}, 
we set $\lambda_i$ as the examination probability $P_{i}$ at the $i^{th}$ rank when computing \textit{DCG}.
Furthermore, based on \textit{DCG}, we can measure multiple ranklists with the cumulative NDCG,
\begin{equation}
\begin{split}
    \textit{eff.}@k_c=\textit{cum-NDCG}@k_c(q,t)&=\sum_{\tau=1}^t \gamma^{t-\tau}\frac{\textit{DCG}@k_c(\pi_\tau)}{\textit{DCG}@k_c(\pi^*)}
    \end{split}
    \label{eq:cumuNDCG}
\end{equation}
where $0\le\gamma\le 1$ is a constant discount factor and $t$ is the current time step. $\pi^*$ is the ideal ranklist constructed by ranking items according to true relevance, and we use $\textit{DCG}@k_c(\pi^*)$, referred to as $\textit{IDCG}$, to normalize $\textit{DCG}@k_c(\pi_\tau)$. By ignoring $\gamma$, we can get the average NDCG as,
\begin{equation}
\begin{split}
    \textit{aver-NDCG}@k_c&=\frac{\sum_{\tau=1}^t \textit{DCG}@k_c(\pi_\tau)}{t\textit{DCG}@k_c(\pi^*)}\\
    &=\frac{ \sum_{d\in D(q)} R(d)\big(\sum_{\tau=1}^{t} \sum_{j=1}^{k_c}  P_{j}\mathbbm{1}_{\pi_i[j]==d}\big)}{t\textit{DCG}@k_c(\pi^*)}\\
    &=  \frac{\sum_{d\in D(q)}R(d)  E^t@k_c(d)}{t\textit{DCG}@k_c(\pi^*)}
    \end{split}
\label{eq:NDCG}
\end{equation}
where $E^t@k_c(d)$ is the cumulative exposure at top $k_c$ ranks,
\begin{equation}
    E^t@k_c(d)=\sum_{\tau=1}^{t} \sum_{j=1}^{k_c}  P_{j}\mathbbm{1}_{\pi_i[j]==d}
    \label{eq:expo}
\end{equation}
where $\pi_i[j]$ indicates the $j^{th}$ item in ranklist $\pi_{i}$, $\mathbbm{1}$ is an indicator function which means we only accumulate item $d$'s exposure.
To simplify notations, we use $E^t(d)$ to denote $E^t@k_s(d)$ and $\textit{eff.}$ for $\textit{eff.}@k_s$, where $k_s$ is the ranklist length introduced in Eq.~\ref{eq:selectionBias}. 

\textbf{The Provider-side Utility (Fairness)}:
As items' rankings can have significant effects on their providers' profit, it is important to create a fair ranking environment. 
To evaluate whether exposure is fairly allocated to users, we use the negative exposure disparity between item pairs as fairness measurement~\cite{oosterhuis2021computationally},
\begin{equation}
\vspace{-2pt}
    \textit{unfair.}(q,t)=\frac{1}{n(n-1)}\sum_{d_x\in D(q)} \sum_{d_y\in D(q)}\bigg(E^t(d_x)R(d_y)-E^t(d_y)R(d_x)\bigg)^2
    \label{eq:unfairness}
\end{equation}
\begin{equation}
    \textit{fair.}(q,t)=-\textit{unfair.}(q,t)
    \label{eq:fairness}
\vspace{-2pt}
\end{equation}
where $D(q)$ is the set of candidate items that will construct the ranklist of query $q$ and $n=|D(q)|$. The intuition of the above fairness measurement is that the optimal fairness can be achieved when items exposure is proportional to their relevance, i.e., $\forall d_x,d_y\in D(q),\:\frac{E^t(d_x)}{R(d_x)}=\frac{E^t(d_y)}{R(d_y)}$. In other words, exposure fairness means we should let items of similar relevance get similar exposure. 
In this paper, we choose the exposure fairness evaluation proposed by \citet{oosterhuis2021computationally}  instead of the original evaluation proposed in~\cite{singh2018fairness}. The reason for the choice is that the fairness evaluation in \cite{singh2018fairness} needs to divide the exposure of an item by its
relevance, i.e., $\frac{E^t(d)}{R(d)}$, which has zero denominator problem when item $d$ is irrelevant, and $R(d)$ is near zero. This paper uses average unfairness across different queries to evaluate a ranking algorithm. We also refer to the average unfairness as the unfairness tolerance. 

\section{PROPOSED METHOD}
Most existing fair algorithms are greedy algorithms, i.e., they sequentially construct the locally optimal ranklist for the next immediate session. 
Therefore they usually fail to optimize the construction procedure if we expect multiple sessions will come for the same query in the future.
To mitigate this gap and reach a global optimal for a query, we propose to 
(i) \textbf{plan ahead and precompute multiple ranklists for future use (\S~\ref{sec:rankingObj}) }and 
(ii) \textbf{jointly optimize those ranklists together to maximize both fairness and ranking relevance (\S~\ref{sec:phase1} \& \S~\ref{sec:phase2})}. 
We hypothesize that jointly optimizing multiple ranklists can construct better ranklists compared to sequentially greedily optimizing one single ranklist at each time step.
Such hypothesis is verified by both the theoretical analysis in \S~\ref{sec:theoreticalAnalysis} and the empirical results in \S~\ref{sec:results}.

\subsection{Future-aware Ranking Objective}
\label{sec:rankingObj}
We first propose a ranking fairness objective to plan and optimize the future $\Delta T$ ranklists for a query $q$. Specifically, when we are at time step $t+1$, the objective is to pre-compute the optimal ranklists $\mathcal{B}^*=[\pi_{t+1},...,\pi_{t+\Delta T}]$ that can maximize the marginal fairness $\Delta \reallywidehat{\textit{fair.}}(q,t,t+\Delta T)$,
\begin{equation}
     \mathcal{B}^*= \argmax_{\mathcal{B}=[\pi_{t+1},...,\pi_{t+\Delta T}]} \quad \Delta\reallywidehat{\textit{fair.}}(q,t,t+\Delta T)    
\label{eq:fairmax}
\end{equation}
$\Delta \reallywidehat{\textit{fair.}}(q,t,t+\Delta T)=\reallywidehat{\textit{fair.}}(q,t+\Delta T) -\reallywidehat{\textit{fair.}}(q,t)$, and $\reallywidehat{\textit{fair.}}$ is the estimated fairness. $\reallywidehat{\textit{fair.}}$ is calculated with Eq.~\ref{eq:fairness} by substituting true relevance $R$ with the estimated relevance, denoted as $\reallywidehat{R}$, since true relevance is mostly not available during optimization. 
Here maximizing the marginal fairness $\Delta\reallywidehat{\textit{fair.}}(q,t,t+\Delta T)$ is equivalent to maximizing final fairness $\reallywidehat{\textit{fair.}}(q,t+\Delta T)$ at time step $t+\Delta T$. The reason for the equivalence is that ranklists $\mathcal{B}=[\pi_{t+1},...,\pi_{t+\Delta T}]$ can only influence the marginal fairness from timestep $t+1$ to timestep $t+\Delta T$ rather than fairness before timestep $t$. Here, fairness evaluation in Eq.~\ref{eq:fairness} is a direct objective in our method.

To the best of our knowledge, there is no trivial algorithm to get the optimal ranklists $\mathcal{B}^*$ due to the discontinuity  of ranking problem~\cite{liu2009learning}. One example of discontinuity  is that increasing an item's ranking score may not change the output ranklist, and fairness stays the same unless the increased score can surpass another item's score, and fairness will experience a sudden change.
To alleviate the discontinuity of $\mathcal{B}^* \leftarrow\operatorname*{argmax}_B \Delta\reallywidehat{\textit{fair.}}$, we propose a novel two-phase solution path by introducing a continuous variable, $\Delta E$, 
\begin{equation}
 \mathcal{B}^* \xleftarrow[phase 2]{\text{Vertical Allocation}}  \Delta E^*(d)\: \forall d\in D(q)   \xleftarrow[phase 1]{ \operatorname*{argmax}_{\Delta E}} \Delta\reallywidehat{\textit{fair.}} 
    \label{eq:SoluPath}
\end{equation}
where $\Delta E (d)$, also referred to as the planning exposure, is the marginal (or incremental) exposure we plan to assign to item $d$ within the next $\Delta T$ timesteps. $\Delta E^*(d)$ is the optimal marginal exposure. We found that introducing $\Delta E$ helps to effectively maximize $\Delta \reallywidehat{\textit{fair.}}$ in phase 1 of our solution. 
In phase 2, we  construct the optimal ranklists $\mathcal{B}^*$ by allocating the optimal exposure $\Delta E^*$ to each item with a vertical allocation method (more details are in \S\ref{sec:phase2}). 

To get $\Delta E^*$,  we carry out a Taylor series expansion to investigate how future exposure will influence the fairness objective,
\begin{equation}
\begin{split}
\Delta \reallywidehat{\textit{fair.}}&(q,t,t+\Delta T) = \ \sum_{d\in D(q)}\frac{\partial{\reallywidehat{\textit{fair.}}}}{\partial E(d)} \Delta E(d) \\&+\frac{1}{2}\sum_{d_x\in D(q)} \sum_{d_y\in D(q)}\frac{\partial^2{\reallywidehat{\textit{fair.}}}}{\partial E(d_x) \partial E(d_y)} \Delta E(d_x) \Delta E(d_y)\\
&=\sum_{d\in D(q)} \reallywidehat{G}(d) \Delta E(d)\\&-\frac{1}{2} \sum_{d_x\in D(q)} \sum_{d_y\in D(q)}\reallywidehat{H}(d_x,d_y)\Delta E(d_x) \Delta E(d_y)\\
&= \vec{\reallywidehat{G}}\cdot \Delta \vec{E}-\frac{1}{2} \Delta \vec{E}^T\cdot  \reallywidehat{H}\cdot  \Delta \vec{E}
\end{split}
\label{eq:expansion}
\end{equation}
where $\Delta E(d)=E^{t+\Delta T}(d)-E^{t}(d)$, the exposure increments. $\reallywidehat{G}$ and $\reallywidehat{H}$ are the first and the second order derivative, i.e., the gradient vector and the Hessian matrix, respectively,
\begin{equation*}
\begin{split}
\reallywidehat{G}(d)&= 
      \frac{4}{n(n-1)}\bigg(\reallywidehat{R}(d)\sum_l E^{t}(l)\reallywidehat{R}(l)-E^{t}(d)\sum_h \reallywidehat{R}^2(h)\bigg)\\
\reallywidehat{H}(d_x,d_y)&=\frac{4}{n(n-1)}\bigg((\sum_{d\in D(q)} \reallywidehat{R}^2(d))\mathbbm{1}_{x=y}-\reallywidehat{R}(d_x)\reallywidehat{R}(d_y)\bigg)
\end{split}
\end{equation*}
By observing Eq.~\ref{eq:unfairness}, we could derive two facts about the above second-order expansion in Eq.~\ref{eq:expansion}. 
(i) The above second-order expansion is not an approximation, but equality since (un)fairness in Eq.~\ref{eq:unfairness} is defined as a polynomial of $E$ with a degree of two, and its derivative of order higher than two is zero. 
(ii) Since (un)fairness in Eq.~\ref{eq:unfairness} is defined as a sum of squares, the second order derivative $\reallywidehat{H}$ is semi-definite.  
Being equality, Eq.~\ref{eq:expansion} allows us to correctly estimate future fairness given marginal exposure $\Delta E$ even when we consider a long-term future (large $\Delta T$ and $\Delta E$). 
Based on the correct future fairness estimation, it is possible to find the optimal marginal exposure planning, denoted as $\Delta E^*$, that can maximize future fairness.
Since $\reallywidehat{H}$ is semi-definite, Quadratic Programming-based (QP) optimization is valid to find $\Delta E^*$.  
We give the specific QP problem formulation to find $\Delta E^*$ in \S~\ref{sec:phase1} and leave constructing optimal ranklists $\mathcal{B}^*$ from $\Delta E^*$  in \S~\ref{sec:phase2}. 
\subsection{\textit{Phase 1}: Future Exposure Planning}
\label{sec:phase1}
When giving the QP problem formulation, we noticed that existing ranking fairness optimization usually considers two settings: (i) \textbf{the post-processing setting}~\cite{singh2019policy,singh2018fairness,biega2018equity} where relevance is assumed to be known or well estimated in advance; 
and (ii) \textbf{the online setting}~\cite{yang2021maximizing,morik2020controlling} where fairness is optimized while relevance is still being learned. 
To consider both settings, we first illustrate the QP problem formulation in the post-processing setting in \S~\ref{sec:postPro} and then extend it to work in the online setting in \S~\ref{sec:OnlinePro}. 

\subsubsection{The post-processing setting}. 
\label{sec:postPro}
To get the optimal exposure planning $\Delta E^*$, we propose the following QP formulation with $\Delta E(d)\: \forall d\in D(q)$ as decision variables,
\begin{subequations}
\begin{align}
\max_{\Delta E} \: \quad & \Delta \reallywidehat{\textit{fair.}}(q,t,t+\Delta T) \label{eq:QPobjPost}\\
\textit{s.t.}\quad  \sum_{d\in D(q)} \Delta E(d)&=\sum_{i=1}^{\Delta T} \sum_{j=1}^{k_s} P_{j} \label{eq:QPfutureExpo}\\
\sum_{d\in D(q)} \Delta E(d) \reallywidehat{R}(d)&\ge (1-\alpha) \sum_{i=1}^{\Delta T} \sum_{j=1}^{k_s} P_{j}\reallywidehat{R}(d_{\textit{order}_j})  \label{eq:NDCGconst}\\
\Delta E(d) &\ge 0, \forall d\in D(q)  \label{eq:ExpoGE0}  \\
\Delta E(d) &\le \sum_{i=1}^{\Delta T} P_{1}, \forall d\in D(q) \label{eq:ExpoGErank1}
\end{align}
\label{eq:FARAFormuPost}
\end{subequations}
where $k_s$ is the length of ranklists, $P_j$ is the examining probability at rank $j$.
Eq.\ref{eq:QPfutureExpo} indicates that the sum of items' marginal exposure should equal the sum of the $\Delta T$ ranklists' exposure.
In Eq.\ref{eq:NDCGconst}, we introduce the $\textit{NDCG}$ constraint, where $\reallywidehat{R}(d_{\textit{order}_j})$ means the $j^{th}$ largest estimated relevance.  According to Eq.(\ref{eq:DCG}\&\ref{eq:NDCG}), $\sum_{d\in D(q)} \Delta E(d) \reallywidehat{R}(d)$ indicates the $\textit{DCG}$ and $\sum_{i=1}^{\Delta T} \sum_{j=1}^{k_s} P_{j}\reallywidehat{R}(d_{\textit{order}_j})$ indicates $\textit{IDCG}$. Then, it is straightforward that ($1-\alpha$) indicates the minimum $\textit{NDCG}$ requirement we want to guarantee. In the post-processing setting, relevance is assumed to be given or already well-estimated prior to ranking optimization.
Eq.\ref{eq:ExpoGE0} indicates that an item's marginal exposure $\Delta E(d)$ should be no less than 0. In Eq.\ref{eq:ExpoGErank1},  $\Delta E(d)$ should be more than the accumulation of the first rank's exposure in the $\Delta T$ ranklists because an item should be unique in a ranklist and there are $\Delta T$ ranklists under consideration. Here we assume that the first rank's exposure is the largest and exposure drops from the top to the lower ranks.

\subsubsection{The online setting.} \,
\label{sec:OnlinePro}
In the online setting, ranklists are optimized while relevance is still being learned. How to actively explore items and get more accurate relevance for ranking optimization is critical. 
\citet{yang2022CanClicks} show that a more accurate relevance estimation for an item can be achieved by exposing an item more because more exposure leads to more interaction with users. 
Based on this, we do explorations by setting a minimum exposure requirement for items and propose the following QP formulation,
\begin{subequations}
\begin{align}
\max_{\Delta E} \: \quad \Delta \reallywidehat{\textit{fair.}}(q,t,t+\Delta T)&- \beta\sum_{d\in D(q)} s(d) \label{eq:QPobjOn}\\
\textit{s.t.} \quad \quad   Eq.~& (\ref{eq:QPfutureExpo},\ref{eq:NDCGconst},\ref{eq:ExpoGE0},\ref{eq:ExpoGErank1})\\
 s(d)&\ge 0, \forall d\in D(q)  \label{eq:sGe0} \\
s(d)+\Delta E(d) +E^{t}(d)&\ge E_{min}, \forall d\in D(q)  \label{eq:sGeRequired}
\end{align}
\label{eq:FARAFormuOnline}
\end{subequations}
where $\beta$ indicates the importance of exploration,  $s(d)\: \forall d\in D(q)$ are slack variables  to encourage exploration. In other words, both $s(d)$ and $\Delta E(d)\: \forall d\in D(q)$ are decision variables in the setting.
In Eq.~\ref{eq:sGeRequired}, $E_{min}$ is the minimum exposure requirement, $E^{t}(d)$ is item $d$'s exposure accumulated till time step $t$, and $\Delta E(d)$ is the marginal exposure we plan to allocate to item $d$ within the next $\Delta T$ steps.  $s(d)$ can be interpreted as the additional exposure still needed to satisfy the minimum exposure requirement after the next $\Delta T$ steps. 
When $E^{t}(d)\ge E_{min}$ for item $d$, i.e., minimum exposure requirement is already satisfied for item $d$, it is straightforward that $s(d)$ will be 0 and will not contribute to the ranking objective in Eq.~\ref{eq:QPobjOn}. 
When $E^{t}(d)< E_{min}$, i.e. item $d$ does not meet the minimum exposure requirement, the objective in Eq.~\ref{eq:QPobjOn} will try to minimize $s(d)$. 
In other words, $\Delta{E}(d)$ will be boosted in order to satisfy the constraint in Eq.~\ref{eq:sGeRequired}. With more exposure, item $d$ will be explored more.
In this paper, we refer to the introduction of $s$ as \textbf{Exploration}. And we treat $\beta$ in Eq.~\ref{eq:QPobjOn} and $E_{min}$ as hyper-parameters to control the degree of exploration.

As quadratic programming has been well studied, there are many available existing solvers. In this paper, we use quadratic programming library qpsolvers\footnote{\url{https://pypi.org/project/qpsolvers/}} within python to solve Equation \ref{eq:FARAFormuPost} and Equation \ref{eq:FARAFormuOnline} to get the optimal exposure planning $\Delta E^*$.

\subsection{\textit{Phase 2}: Ranklists Construction}
\label{sec:phase2}
Following the solution path in Eq.~\ref{eq:SoluPath}, the next step is to  construct the optimal ranklists $\mathcal{B}^*$ according to $\Delta E^*$ as  $\Delta E^*$ has been solved in Phase 1. Here, we should allocate each item exactly its optimal exposure $\Delta E^*$ within $\mathcal{B}^*$. However, we find that the allocation solution of $\mathcal{B}$ is not unique, and they share the same aver-NDCG$@k_s$ (see  Theorem \ref{sec:fixedEffFair}). Therefore, we additionally aim to find the optimal $\mathcal{B}^*$ that can optimize \textbf{all top ranks' effectiveness}, i.e., aver-NDCG$@k_c$, $\forall k_c\le k_s$. 
Optimizing top ranks' effectiveness is important since users usually pay more attention to top ranks. 


\begin{figure}
    \centering
    \includegraphics[scale=0.5]{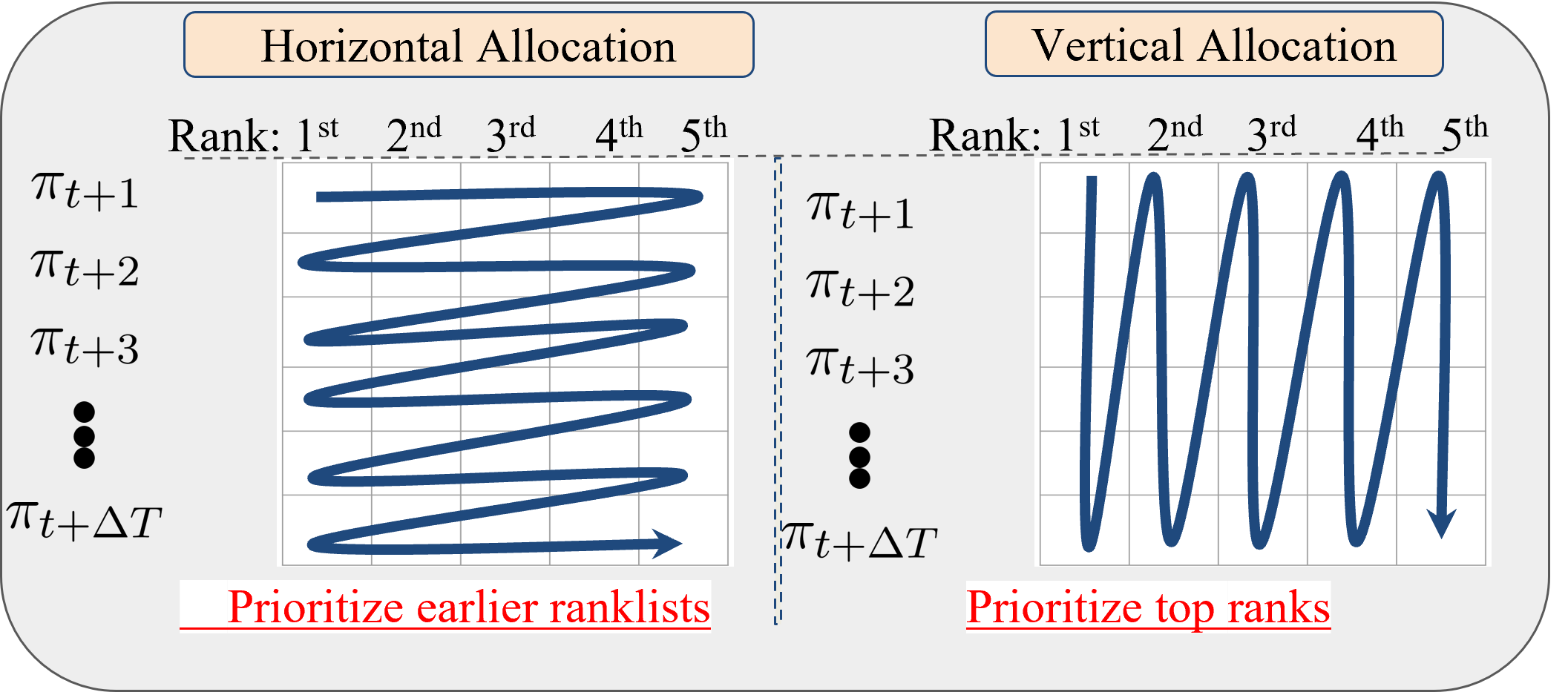}
    \caption{The ranklist construction order of the horizontal allocation and vertical allocation.}
    \label{fig:VHcompare}
\end{figure}

\begin{algorithm}[t]
\textbf{Input}: 
The optimal exposure planning $\Delta E$, the number of planning sessions to consider $\Delta T$, the ranked list length $k_s$, relevance estimation $\reallywidehat{R}$\;
\textbf{Initialize:} ranking lists $\mathcal{B}^*(i)\shortleftarrow$ $[\varnothing]$  for $i \in range (\Delta T)$,
set $\Tilde{E}(d)\shortleftarrow 0\quad \forall d\in D$ \;
\For{$rnk\in [1,2,...,k_s]$}{\label{algoline:loop1}
    \For{$sess \in [1,2,...,\Delta T]$ }{\label{algoline:loop2}
        $Set1 \shortleftarrow \{d|\Delta E(d)-\Tilde{E}(d)\geq P_{rnk}$\}\; \label{algo2:margin}
        $Set2 \shortleftarrow \{d| d \not\in \mathcal{B}^*(sess))$\}\;
        \uIf{$Set1 \cap Set2 =  \varnothing
        $}{$Candidates\shortleftarrow  Set2$}
        \Else{$Candidates\shortleftarrow Set1 \cap Set2$}
        
       $ d^*\shortleftarrow\underset{d\in Candidates}{\arg\!\max} {\reallywidehat{R}(d)}$\;\label{algoline:priority}
       $\mathcal{B}^*(sess).append(d^*)$\;
       $\Tilde{E}(d^*)\shortleftarrow \Tilde{E}(d^*) +P_{rnk}$\;}
            }
\textbf{Output}: $\mathcal{B}^*$;  
\caption{Vertical Allocation}
\label{algo:VerticalAlloc}
\end{algorithm}
Inspired by~\cite{yang2022effective},  we propose a vertical exposure allocation method in Algorithm \ref{algo:VerticalAlloc} to construct the optimal $\mathcal{B}^*$ based on $\Delta E^*$. The difference between a vertical allocation and a horizontal allocation is the ranklist construction order. As shown in Fig.~\ref{fig:VHcompare}, a horizontal allocation prioritizes earlier ranklists and first fills out all ranks of the $i^{th}$ ranklist $\pi_i$ before filling out  $\pi_{i+1}$.  However, a vertical allocation prioritizes top ranks  and fills out the $i^{th}$ ranked items of all ranklists before filling out any $(i+1)^{th}$ ranked item.  
Since top ranks are usually more important, our proposed Algorithm \ref{algo:VerticalAlloc} adopts a vertical allocation to fill out $\mathcal{B}^*$. In our proposed Algorithm \ref{algo:VerticalAlloc}, $\mathcal{B}^*$ is the generated $\Delta T$ ranklists and  $\mathcal{B}^*(sess,rnk)$ denote the ${rnk}^{th}$ rank of the ${sess}^{th}$ ranklist. 
To fill out $\mathcal{B}^*(sess,rnk)$, the proposed vertical allocation first generates a feasible candidate set, i.e., $Candidates$, which contains items that are not selected for this session before ($d \not\in B(sess)$) but still have planned exposure left ($\Delta E(d)-\Tilde{E}(d)\geq P_{rnk}$). 
Here, examination probability $P_{rnk}$ serves as the margin,  $\Tilde{E}(d)$ stores the actual exposure item $d$ receives. Within $Candidates$, our algorithm selects the most relevant item from the candidate set to fill out $\mathcal{B}^*(sess,rnk)$. Algorithm \ref{algo:VerticalAlloc} is theoretically justified to accurately allocate exposure $\Delta E^*$ (see  Theorem \ref{sec:ExposureError}) and can construct optimal $\mathcal{B}^*$ for aver-NDCG$@k_c$  $\forall k_c\le k_s$ (see  Theorem\S~\ref{sec:VerticalReason}). 

Although inspired by the vertical method in \cite{yang2022effective}, the proposed vertical allocation is different from it. \citet{yang2022effective} focus on a certain share of exposure to be guaranteed and have a complicated 3-step procedure, i.e., allocation, appending, and resorting, which cannot be used to allocate $\Delta E^*$ for our problem.  
The proposed allocation algorithm in this paper uses up all $\Delta E^*$, and the allocation procedure is less complicated and more straightforward than those introduced by \citet{yang2022effective}. 

\begin{algorithm}[t]
\caption{FARA: Future-aware Ranking Algorithm}\label{algo:FARA}
 \textbf{Input}:  
   The number of planning sessions to consider $\Delta T$, fairness-effectiveness tradeoff parameter $\alpha$. And in the online setting, we need to additionally give the exploration parameters $\beta$ and $E_{min}$ (see Eq.~\ref{eq:FARAFormuOnline}) \;
  \textbf{Initialize} time step $t \gets 0$, initialize an empty dictionary $\mathbbm{B}=\{\}$ to store ranked lists, items' exposure $E\gets 0$ and cumulative click $\textit{cumC} \gets 0$\;
\While{True}{
   $t \gets t+1$\;
  A user issues a query $q_t$\;
\If{$q_t \not\in \mathbbm{B}$}{$\mathbbm{B}[q_t]=[]$, i.e., add an empty list}\
   \If{$\mathbbm{B}[q_t]$ is empty}{
   Get $\Delta E^*$ by solving Eq.~\ref{eq:FARAFormuPost} (post-processing) or Eq.~\ref{eq:FARAFormuOnline} (online)\;
   Get $\mathcal{B}^*$ with Algorithm~\ref{algo:VerticalAlloc}\;
   Randomly shuffle $\mathcal{B}$\;
   $\mathbbm{B}[q_t] \gets \mathcal{B}^*$;}
   
   Pop out a ranking list from $\mathbbm{B}[q_t]$ and present it\;
   Update cumulative click $cumC$ and  items' exposure $E$\;
   Update the relevance estimation \reallywidehat{R} via Eq.~\ref{eq:releEsti}\;
}
\end{algorithm}
\subsection{FARA: Future-Aware Ranking Algorithm}
Combining \textit{Phase 1} and \textit{Phase 2}, we propose a future-aware ranking algorithm for fairness optimization, \textbf{FARA}, detailed in Algorithm~\ref{algo:FARA}. 
FARA serves users in an online manner where we pre-compute $\mathcal{B}^*$, the next $\Delta T$ ranklists, and randomly pop out one ranklist from $\mathcal{B}^*$ when needed. 
When $\mathcal{B}^*$ is used up and empty, We will re-compute $\mathcal{B}^*$ for future $\Delta T$ timesteps. $\mathbbm{B}[q]$ is used to store $\mathcal{B}^*$ for query $q$.
Besides, FARA does not depend on any specific relevance estimation model, therefore, can
be seamlessly integrated into most existing ranking applications.
In this paper, we follow works by~\cite{yang2022CanClicks} to use the following unbiased estimator of relevance,
\begin{equation}
    \reallywidehat{R}(d)=\frac{\textit{cumC}^t(d)}{E^t(d)}
\label{eq:releEsti}
\end{equation}
where $\textit{cumC}^t(d)= \sum_{i=1}^t \sum_{j=1}^{k_s} C_{i,j}\mathbbm{1}_{\pi_i[j]==d} $ is the cumulative clicks\footnote{Given the page limit, we skip the proof of unbiasedness for the above estimator and refer interested readers to \cite{yang2022CanClicks}}.
Moreover, it is worth noting that the relevance estimator can be replaced with other relevance estimators as well.


\section{Theoretical Analysis}
\label{sec:theoreticalAnalysis}
\begin{theorem}
    Algorithm \ref{algo:VerticalAlloc} can theoretically guarantee $\Delta E(d)-\Tilde{E}(d)\le P_{k_s}$ for at least $|D|-k_s$ items.
\label{sec:ExposureError}
\vspace{-5pt}
\end{theorem}

\begin{proof}
Here we discuss the exposure allocation error bounds in Phase 2 of FARA, i.e., $|\Tilde{E}(d)-\Delta E(d)|$, in Algorithm~\ref{algo:VerticalAlloc}. 
We noticed that there are two possible scenarios of lines 6-12 in Algorithm~\ref{algo:VerticalAlloc}: 

\textit{Scenario 1: There exists a $(rnk^*, sess^*)$ pair where $Set1 \cap Set2 =  \varnothing$}. If this happens, $(k_s, sess^*)$ will also have $Set1 \cap Set2 =  \varnothing$ since the size of $Set1$ and $Set2$ monotonically decrease  for lower rank of the same session. 
As $Set2$ is the set of unselected items for a session, we know that there are at least $|D|-k_s$ items in $Set2$, i.e., $|Set2|\ge |D|-k_s$.  If $Set1 \cap Set2 =  \varnothing$, those $|D|-k_s$ items are not in $Set1$.  In other words, there are at least $|D|-k_s$ items that satisfy $\Delta E(d)-\Tilde{E}(d)< P_{k_s}$.


\textit{Scenario 2: $Set1 \cap Set2 \neq \varnothing \: \forall (rnk, sess) $ pair}. Due to the margin $P_{rnk}$ in line~\ref{algo2:margin} of Algorithm~\ref{algo:VerticalAlloc}, $\Delta E(d)\ge \Tilde{E}(d) \: \forall d \in D$ should always hold if $Set1 \cap Set2 = \varnothing$ never happens.  By considering $\Delta E(d)\ge \Tilde{E}(d) \: \forall d \in D$ and the identity $\sum_{d\in D(q)}\Delta E(d)\equiv\sum_{d\in D(q)} \Tilde{E}(d)$, we would know that $\Delta E(d)= \Tilde{E}(d) \: \forall d \in D$, which means exposures are perfectly allocated according to $\Delta E(d)$.

Combing the two scenarios, the vertical allocation in Algorithm \ref{algo:VerticalAlloc} can theoretically guarantee $\Delta E(d)-\Tilde{E}(d)\le P_{k_s}$ for at least $|D|-k_s$ items. 
Since $\Delta E(d)>> P_{k_s}$ and $|D|>>k_s$ when using FARA, we can claim that Algorithm \ref{algo:VerticalAlloc} correctly allocates exposure.
\end{proof}

\begin{theorem}
\vspace{-5pt}
FARA can reach the optimal NDCG with the given exposure planning.
\label{sec:VerticalReason}
\vspace{-5pt}
\end{theorem}
\begin{proof}
Here we provide theoretical proof that vertical allocation, i.e., phase 2, can optimize effectiveness (\textit{aver-NDCG}) when exposure planning $\Delta E$ is given. 
Specifically, maximizing \textit{aver-NDCG}$@k_c$ in Eq.~\ref{eq:NDCG} is equivalent to
\begin{subequations}
\begin{align}
\vspace{-15pt}
\max \sum_{d\in D(q)}R(d)&  E@k_c(d) \label{eq:ignoreNorm} \\
\textit{s.t.}\quad \sum_{d\in D(q)}  E@k_c(d)&=Const.\label{eq:ExpoSumConstkc} \\
0\le E@k_c(d)&\le \Delta E(d)\; \: \forall d\in D \label{eq:LessThanExpoPlanning}
\vspace{-5pt}
\end{align}
\label{eq:proofOptim}
\end{subequations}
where normalization is ignored in Eq.~\ref{eq:ignoreNorm}, the sum of top ranks exposure should be a constant in Eq.~\ref{eq:ExpoSumConstkc}, and the top ranks exposure should be less than the total exposure planning in Eq.~\ref{eq:LessThanExpoPlanning}. According to Rearrangement Inequality~\cite{hardy1952inequalities}, it is straightforward to know that \textit{aver-NDCG}$@k_c$, i.e., Eq.~\ref{eq:ignoreNorm},  can be optimized by letting item of greater relevance $R$ get more  exposure at top ranks, i.e., greater $E@k_c$. In other words, we should prioritize letting items of greater relevance $R$ fulfill their exposure planning $\Delta E$ at top $k_c$ ranks since $E@k_c$ is bounded in $[0,\Delta E]$. 
By assuming that \textit{aver-NDCG} at higher ranks is more important~\cite{robertson1977prp}, we should maximize \textit{aver-NDCG}$@k_c$ before maximizing \textit{aver-NDCG}$@(k_c+1)$, $\: \forall\: 1\le k_c<k_s$.
As we maximize \textit{aver-NDCG} from top to lower ranks, it is straightforward that the optimal way is to follow a greedy selection strategy to let an item of greater relevance $R$ fulfill its exposure planning $\Delta E$ at its highest possible ranks. In Algorithm \ref{algo:VerticalAlloc}, the proposed vertical allocation exactly follows the above greedy selection strategy to let item of greater relevance $R$ (line ~\ref{algoline:priority} in Algorithm \ref{algo:VerticalAlloc}) fulfill its exposure planning $\Delta E$ at the highest possible ranks (setting rank loops as the outer loop in ~\ref{algoline:loop1} and line ~\ref{algoline:loop2} of Algorithm \ref{algo:VerticalAlloc}).  So it can reach optimal effectiveness at the top ranks.  
\end{proof}
\begin{theorem}
\vspace{-5pt}
Effectiveness and fairness are fixed when $\Delta E$ is fixed.
\label{sec:fixedEffFair}
\vspace{-5pt}
\end{theorem}
\begin{proof}
\vspace{-10pt}
Given the same exposure planning $\Delta E$, effectiveness (aver-NDCG$@k_s$) and fairness are fixed since we can substitute exposure planning $\Delta E(d)$  for $E^t@k_c(d)$ in Eq.~\ref{eq:NDCG} and substituting $\Delta E$  for $E$ in Eq.~\ref{eq:unfairness}, respectively. In other words, for any ranklist $\mathcal{B}^*$, as long as exposure planning $\Delta E$ can be accurately allocated in $\mathcal{B}^*$, the effectiveness and fairness are fixed.
\end{proof}
\vspace{-8pt}
\begin{table}[t]
	\centering
	\vspace{-5pt}
	\caption{Datasets statistics. For each dataset, the table below shows the number of queries, the average number of docs for each query, and the relevance annotation $y$'s range.}
	\scalebox{1.0}{
		\begin{tabular}{c c c c} \toprule
			Datasets & \#Queries & \#Aver. Docs per Query &$y$'s range \\ \midrule
			MQ2008 & 800  & 20 &$0-2$\\ 
			MSLR-10k& 10k & 122 &$0-4$\\ 
			Istella-S & 33k & 103&$0-4$\\   
			\bottomrule
	\end{tabular}}
	\label{tab:data_statistics}
\end{table}

\section{EXPERIMENTS}
\label{sec:experimental_results}
\vspace{-5pt}
\subsection{Experimental setup}
\vspace{-3pt}
\noindent
\textbf{\textit{Datasets:}} \space
In this work, we use three public Learning-to-Rank (LTR) datasets: MQ2008~\cite{qin2013introducing}, MSLR10k\footnote{\url{https://www.microsoft.com/en-us/research/project/mslr/}} and Istella-S~\cite{lucchese2016post}. 
Datasets' statistics are shown in Table~\ref{tab:data_statistics}.
MQ2008 has a three-level relevance judgment (from 0 to 2). MSLR10k and Istella-S have a five-level relevance judgment (from 0 to 4). 
Queries in each dataset are already divided into training, validation, and test partitions according to a 60\%-20\%-20\% scheme. In this work, we mainly focus on comparison within the LTR tasks. However, the proposed method can be adapted to recommendation tasks, which we leave for future studies.

\vspace{2pt}
\noindent
\textbf{\textit{Baselines:}} \space
\label{sec:baselines}
In this paper, we compare the following methods:\\
\textbullet \, \textbf{TopK}: Sort items according to $\reallywidehat{R}(d)$\\
\textbullet \, \textbf{RandomK}: Randomly rank items.\\
\textbullet \, \textbf{FairCo}~\cite{morik2020controlling}: Fair ranking algorithm based on a proportional controller. $\alpha \in [0.0,1000.0]$ \\
\textbullet \, \textbf{MCFair}~\cite{yang2023marginal}: Fair ranking algorithm directly uses gradient as the ranking score. $\alpha \in [0.0,1000.0]$ \\
\textbullet \, \textbf{ILP}~\cite{biega2018equity}: Fair ranking algorithm based on Integer Linear Programming (ILP).$\alpha \in [0.0,1.0]$ \\
\textbullet \, \textbf{LP}~\cite{singh2018fairness}: Fair ranking  algorithm based on Linear Programming (LP).$\alpha \in [0.0,1000.0]$ \\
\textbullet \, \textbf{MMF}~\cite{yang2021maximizing}. Similar to FairCo but focus on top ranks fairness. $\alpha \in [0.0,1.0]$ \\
\textbullet \, \textbf{PLFair}~\cite{oosterhuis2021computationally}. A fair ranking  algorithm based on Placket-Luce optimization. $\alpha \in [0.0,1.0]$ \\
\textbullet \, \textbf{FARA-Horiz.} (ours): A variant of FARA. Compared to FARA, we switch line \ref{algoline:loop1} and line \ref{algoline:loop2} in Algorithm~\ref{algo:VerticalAlloc} to first iterate the sessions and then iterate the ranks. We refer to the iterations as the \textbf{horizontal allocation} paradigm. $\alpha \in [0.0,1.0]$ \\
\textbullet \, \textbf{FARA} (ours). The proposed fair ranking algorithm. $\alpha \in [0.0,1.0]$.

\vspace{2pt}
\noindent
Among the above ranking algorithm, TopK and RandomK are unfair algorithms, while the others are fair algorithms. While all the fair ranking algorithms aim to maximize effectiveness and fairness, FARA and FARA-Horiz. differ from others by taking a joint optimization across multiple ranklists rather than a traditional greedy optimization approach.
For fair ranking algorithms, there exists a tradeoff parameter $\alpha$, similar to $\alpha$ in Eq.~\ref{eq:FARAFormuPost},  to balance effectiveness and fairness. For fair algorithms, the greater $\alpha$ is, the more we care about fairness
while potentially sacrificing more effectiveness. For example, when increasing $\alpha$ in Eq. \ref{eq:NDCGconst}, FARA can maximize fairness with less effectiveness constraint.
For different fair algorithms, $\alpha$ lies in different ranges.  
For FairCo, MCFair, LP, $\alpha$ are originally within $[0.0,+\infty]$, and we adopt $\alpha\in [0.0,1000.0]$ which is enough according to our experiments. 
For ILP, MMF, PLFair, FARA-Horiz. and FARA, $\alpha\in [0.0,1.0]$.
Although the vertical allocation in Algorithm~\ref{algo:VerticalAlloc} was inspired by \cite{yang2022effective}, \cite{yang2022effective} cannot be used as a baseline because 
\cite{yang2022effective} works with offline ranking services where all user queries are known in advance. However, in this paper, we consider the online services depicted in Algorithm~\ref{algo:FARA}. 

\vspace{3pt}
\noindent
\textbf{\textit{Ranking Service Simulation:}} \space
\label{sec:Simulation}
Following the workflow in Algorithm~\ref{algo:FARA}, at each time step, a simulated user will issue a query $q$, which is randomly sampled from the training, validation, or test partition. Corresponding to the query $q$, a ranking algorithm will construct a ranklist $\pi$ of candidate items and present it to the simulated user. To collect users' feedback for the  ranked list $\pi$, we need to simulate relevance and examination (see Eq.\ref{eq:clickProb}). Same as~\cite{ai2018unbiased}, the relevance probabilities of each document-query pair $(d,q)$ are simulated with their
relevance judgement $y$ as  
$$
\vspace{-5pt}
P(r=1|d,q)=\epsilon+(1-\epsilon)\frac{2^y-1}{2^{y_{max}}-1}
\vspace{-2pt}
$$ 
where $y_{max}$ is the maximum value of relevance judgement $y$, i.e., 2 or 4 depending on the datasets. Besides relevance, following~\cite{morik2020controlling,oosterhuis2021unifying}, we simulate users' examination probability as,
\begin{equation*}
\vspace{-5pt}
    P(e=1|d,\pi)=\begin{cases}
      \frac{1}{\log_2(\textit{rank}(d|\pi)+1)}, & \text{if}\quad \textit{rank}(d|\pi)\le k_s \\
      0, & \text{otherwise}
    \end{cases}
\vspace{-2pt}
\end{equation*}
For simplicity, we only simulate users' examination behavior on top ranks, and we set $k_s$ to 5 throughout the experiments (refer to Eq.~\ref{eq:selectionBias} for more details of $k_s$). With $P(r=1|d,q,\pi)$ and $P(e=1|d,\pi)$, we sample clicks with Equation~\ref{eq:clickProb}.  
The advantage of the simulation is that it allows us to do online experiments on a large scale while still being easy to reproduce by researchers without access to live ranking systems~\cite{oosterhuis2021unifying}.
For simplicity, same as existing works~\cite{oosterhuis2021unifying,yang2022CanClicks,yang2021maximizing,morik2020controlling}, we  assume  that users' examination $P(e=1|d,\pi)$ is known in experiment since  many existing works \cite{ai2018unbiased,wang2018position,agarwal2019estimating,radlinski2006minimally} have been proposed to estimate it. Due to different data sizes, we simulate 200k steps for MQ2008 and  4M  steps for  MSLR10k and Istella-S. 

\vspace{3pt}
\noindent
\textbf{\textit{Experiment Settings:}} \space
We noticed that LP and ILP methods are proposed in \textbf{the post-processing setting}, where relevance is already known or well estimated in advance. However, in most real-world settings, ranking optimization and relevance learning are carried out at the same time, which we refer to as \textbf{the online setting}. 
To give a comprehensive comparison, we 
evaluate ranking methods in both settings. 
In the post-processing setting, all the ranking methods in Section~\ref{sec:baselines} are based on true relevance $R$, and FARA will set $\beta$ as 0. 
In the online setting, all the ranking methods in Section~\ref{sec:baselines} are based on the relevance estimation $\reallywidehat{R}$ in Eq. \ref{eq:releEsti} to perform ranking optimization. FARA set $\beta$ to 1 and $E_{min}=10$ unless otherwise explicitly specified, as they work well across all our experiments.  

\begin{figure*}[t]
\centering
\begin{subfigure}[]{0.33\textwidth}
\includegraphics[scale=0.20]{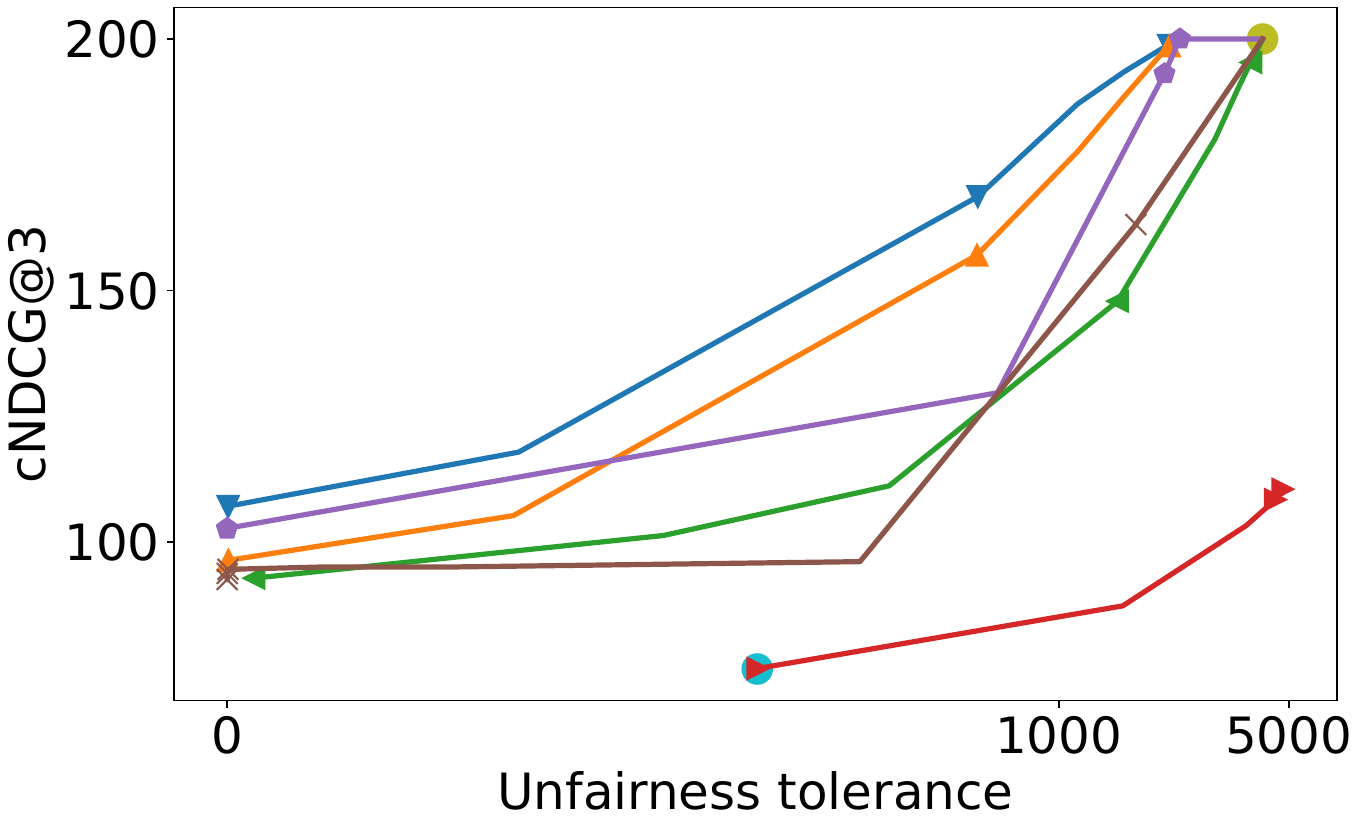}
\caption{MSLR10k, post-processing}
\label{fig:tradeoff1MSPost}
\end{subfigure}\hfill
    \begin{subfigure}[]{0.33\textwidth}
\includegraphics[scale=0.20]{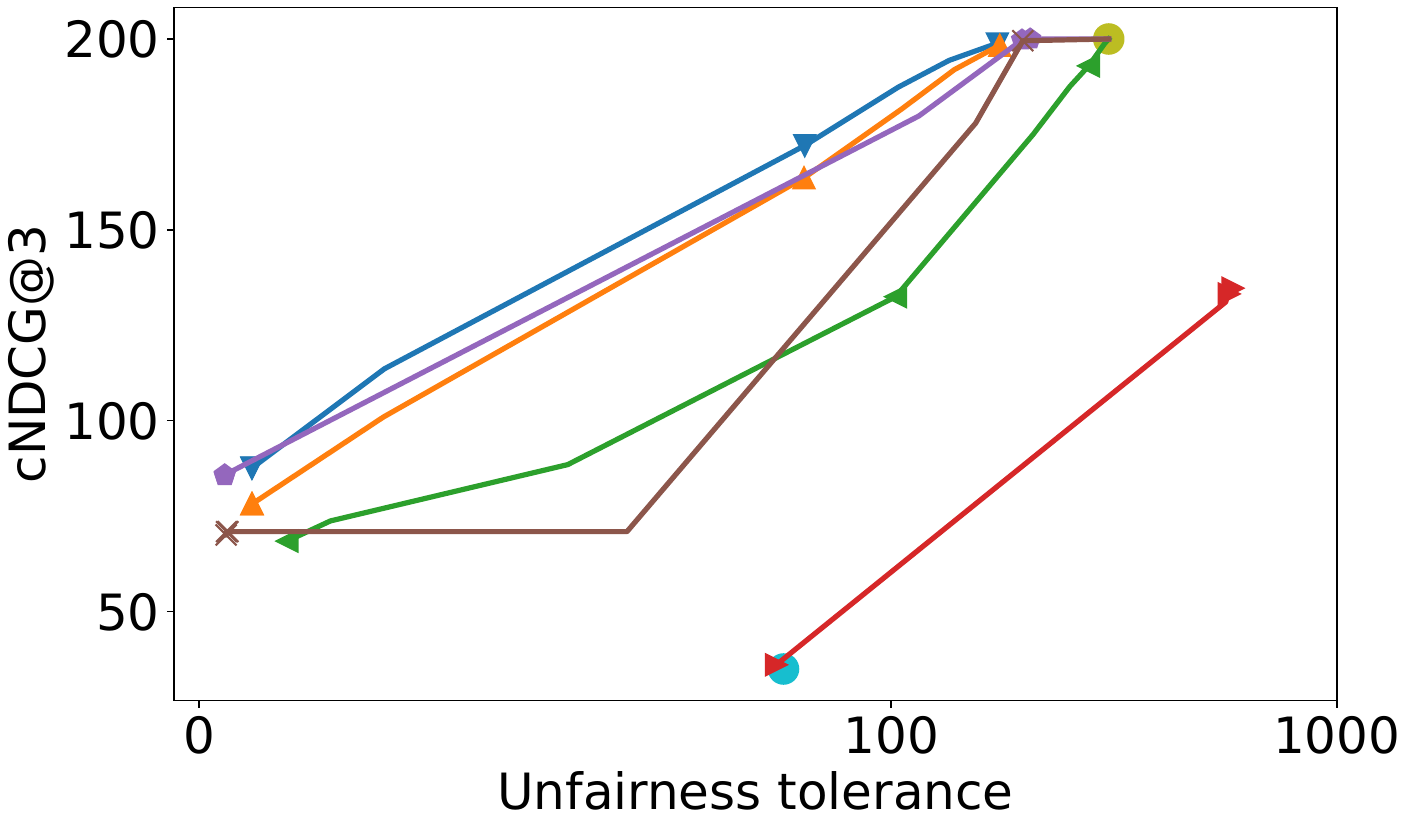}
\caption{Istella-S, post-processing}
\label{fig:tradeoff1MSPost}
\end{subfigure}\hfill
\begin{subfigure}[]{0.33\textwidth}
\includegraphics[scale=0.20]{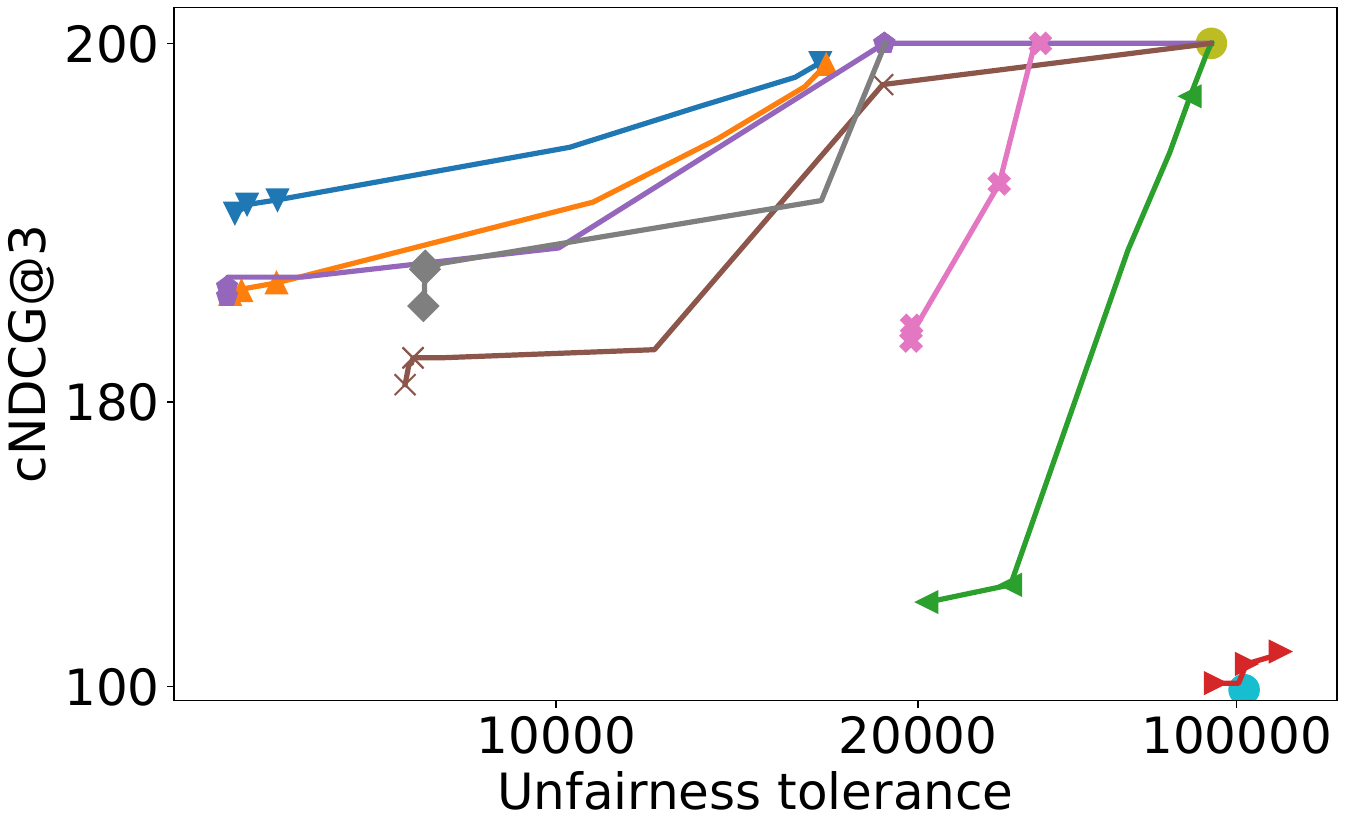}
\caption{MQ2008, post-processing}
\label{fig:tradeoff1MQPost}
\end{subfigure}


%
\begin{subfigure}[]{0.33\textwidth}
\includegraphics[scale=0.20]{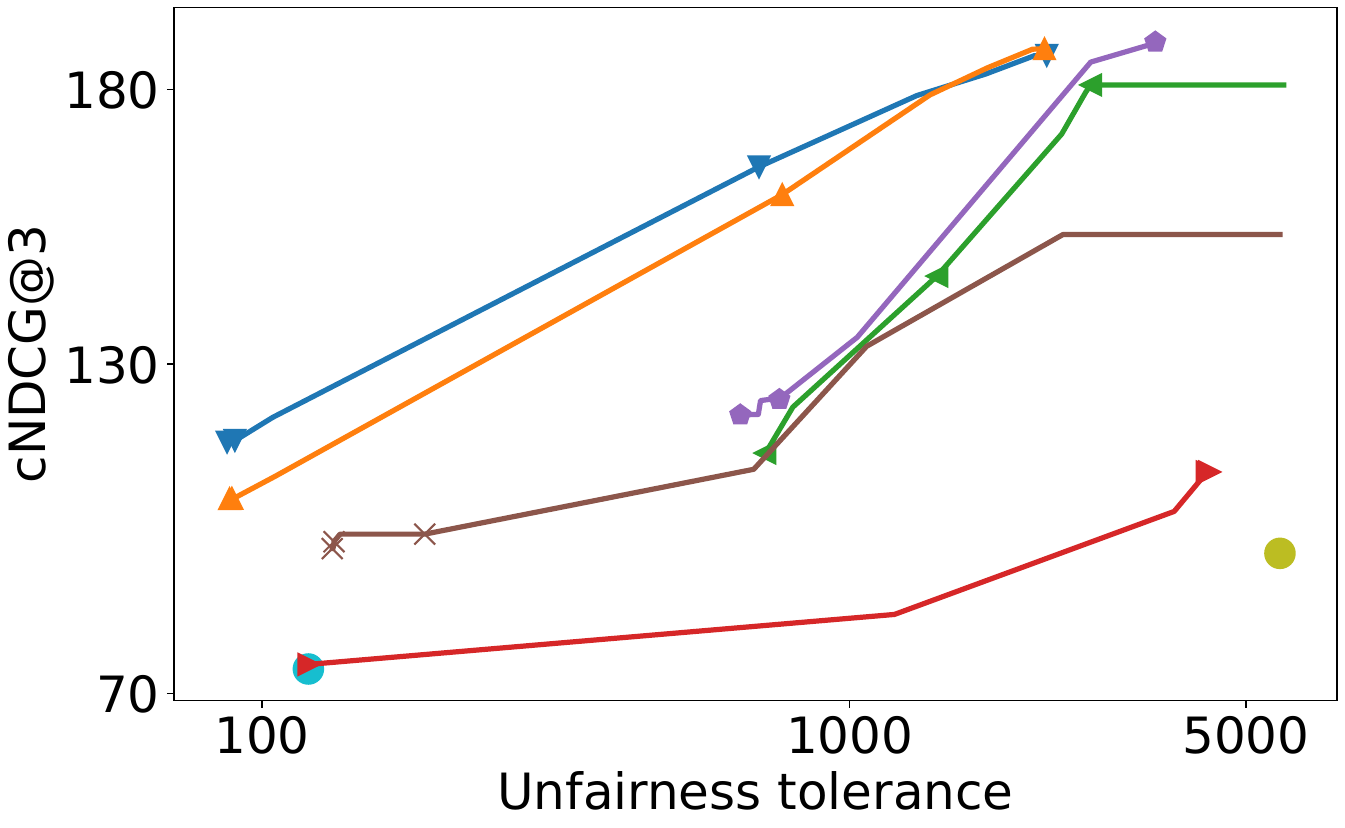}
\caption{MSLR10k, online}
\label{fig:tradeoff1MSOnline}
\end{subfigure}\hfill
\begin{subfigure}[]{0.33\textwidth}
\includegraphics[scale=0.20]{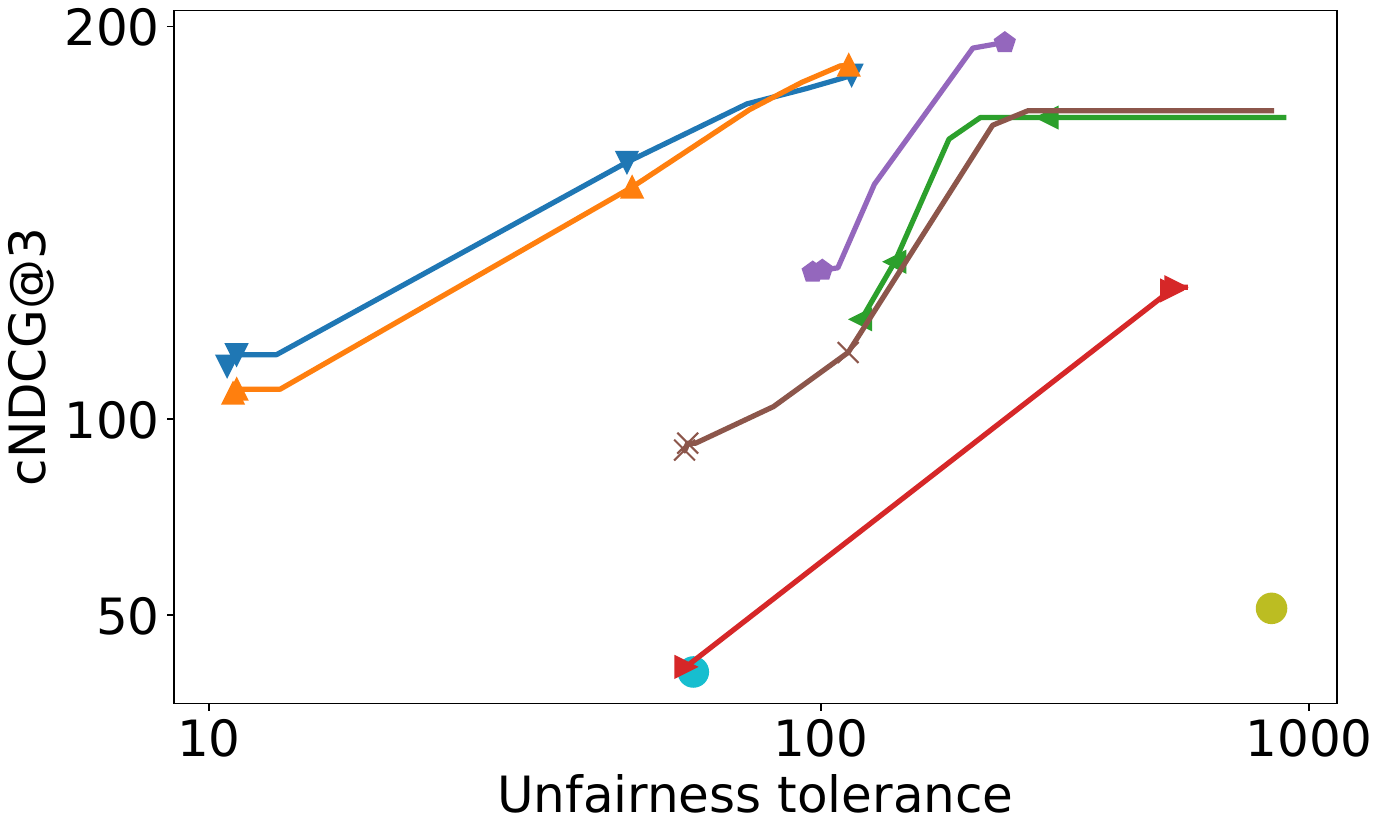}
\caption{Istella-S, online}
\label{fig:tradeoff1MSOnline}
\end{subfigure}\hfill
\begin{subfigure}[]{0.33\textwidth}
\includegraphics[scale=0.20]{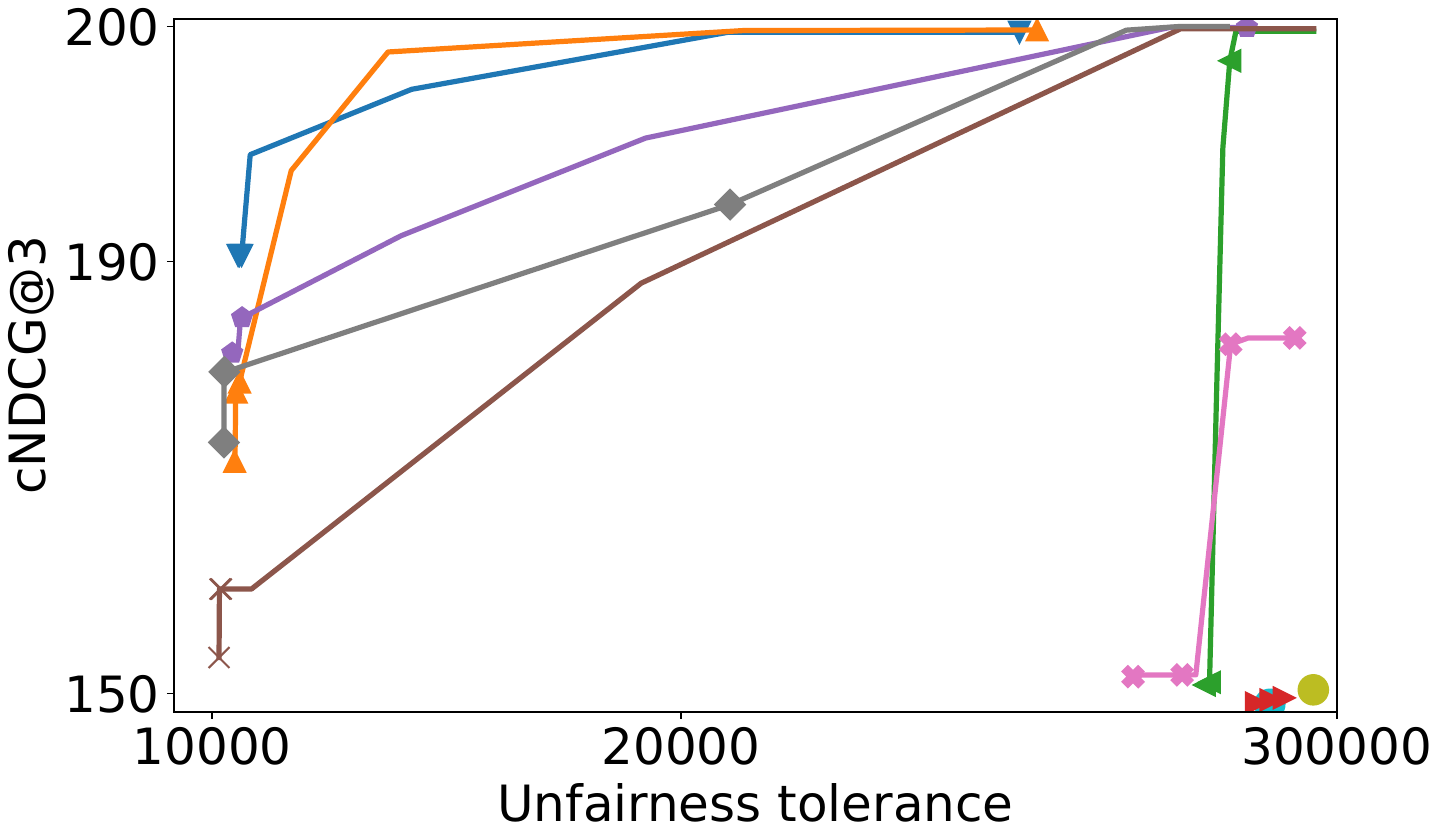}
\caption{MQ2008, online}
\label{fig:tradeoff1MQOnline}
\end{subfigure}
\begin{subfigure}[]{0.99\textwidth}
\vspace{-3pt} 
\centering
\includegraphics[scale=0.3]{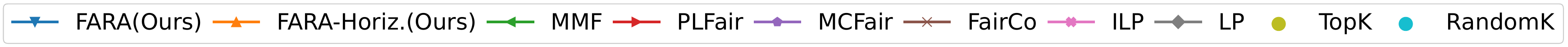}
\label{fig:legend}
\end{subfigure}
\vspace{-10pt} 
\caption{c-NDCG vs. unfairness tolerance (Eq.~\ref{eq:unfairness}) in the post-processing setting and the online setting.  Given the same unfairness, the higher curves or points lie, the better their performances are. Our methods FARA and FARA-Horiz. lie higher than  all  fair baselines in all figures. ILP
and LP are unavailable for MSLR10k and Istella-S due to time costs (refer to Table~\ref{tab:Timeperformance}).}
\label{fig:balance}
\vspace{-15pt}
\end{figure*}

\vspace{3pt}
\noindent
\textbf{\textit{Evaluation:}} \space 
We use the cum-NDCG (cNDCG) in Eq.~\ref{eq:cumuNDCG} with $\gamma=0.995$  (same $\gamma$ adopted in \cite{wang2018efficient,wang2019variance}) to evaluate the effectiveness at different cutoffs, $1\le k_c\le 5$.
Aside from effectiveness,  unfairness defined in Eq.~\ref{eq:unfairness} is used for unfairness measurement. We run each experiment five times and report the average evaluation performance on the test partition. We use the Fisher randomization test~\cite{smucker2007comparison} with $p<0.05$ to do significant tests. Due to the time cost (see Table~\ref{tab:performance}), we do not run ILP and LP on the larger datasets,  MSLR10k and Istella-S, and the performances are not available (NA).


\vspace{-5pt}
\subsection{Results and Analysis}
\label{sec:results}

\begin{table*}[t]
\centering
\caption{
Comparison of cNDCG@(1,3,5) and unfairness tolerance in the post-processing setting. Significant improvements or degradations with respect to FairCo are indicated with +/-. 
Within fair algorithms, the best performance with statistical significance is bolded and underlined. Here, $\alpha$ is set  to the maximum value (see Sec.~\ref{sec:Simulation} for $\alpha$'s range) for each fair algorithm respectively, which means that all algorithms are trying their best to optimize ranking fairness (Eq.~\ref{eq:fairness}) and the numbers in the table represents their unfairness lower bound.
Results are rounded to one decimal place.
}
\centering
\resizebox{\textwidth}{!}{
\begin{tabular}{c|cccc|cccc|cccc}\toprule
\multirow{2}{*}{Methods} & \multicolumn{4}{c|}{MSLR-10k}&\multicolumn{4}{c|}{Istella-S} &\multicolumn{4}{c}{MQ2008} \\ \cline{2-5} \cline{6-9} \cline{10-13}
&cNDCG@1&cNDCG@3&cNDCG@5& unfair. &cN@1&cN@3&cN@5& unfair.&cN@1&cN@3&cN@5& unfair. \\ \hline
TopK&200.0$^+$&200.0$^+$&200.0$^+$	&4165.0	$^-$&200.0$^+$	&200.0$^+$&200.0$^+$	&310.1$^-$	&200.0$^+$&200.0$^+$&200.0$^+$&86001.1	$^-$ \\ 
Randomk&68.0$^-$ &74.7$^-$ &79.7$^-$&119.0	$^-$
&30.2$^-$ &35.7	$^-$&41.1 $^-$&56.7	$^-$
&74.0	$^-$&95.7$^-$ &114.6$^-$ &104632.2$^-$ \\ \hline
PLFair&68.2	$^-$ &74.8	$^-$&79.9$^-$&119.6	$^-$ &31.8 $^-$&36.2$^-$ &41.5 $^-$&54.6 $^-$&79.2$^-$	&99.3$^-$ &117.2 $^-$&101245.1$^-$ \\ 
MMF&84.4&92.8 &99.9	&8.0$^-$ & 62.2&68.8	&80.1	&6.6$^-$&132.8$^-$	&162.3$^-$	&172.5$^-$	&20688.7$^-$ \\ 
ILP&NA&NA&NA&NA&NA&NA&NA&NA&185.6$^+$&183.8	&186.7	&19916.3	$^-$\\
LP&NA&NA&NA&NA&NA&NA&NA&NA&188.4$^+$		&187.4$^+$	&187.9	&9425.7 \\ 
MCFair &114.8$^+$		&102.7$^+$		&101.0	&0.0
 &113.7	$^+$	&85.25	$^+$	&81.3	&0.4
    &193.5$^+$	&186.0$^+$&186.6	&9113.7	
\\
FairCo&85.5	&93.7	&100.8	&0.0&63.3&69.9&80.4	&0.5	&179.0	&182.0	&187.4	&9382.0	\\ \hline
FARA-Horiz.(Ours)&90.7$^+$		&96.1$^+$	&100.4	&0.0	 &78.5$^+$&79.8$^+$	&82.6	&0.9	&187.3$^+$		&186.1$^+$		&187.0	&9125.9	\\ 
FARA(Ours)&\underline{\textbf{129.0}}$^+$	&\underline{\textbf{107.0}}$^+$	&99.7	&0.0	&\underline{\textbf{135.5}}$^+$	&\underline{\textbf{89.1}}$^+$&82.6	&0.9	&\underline{\textbf{196.3}}$^+$	&\underline{\textbf{190.9}}$^+$	&186.8	&9129.9	\\
\bottomrule
\end{tabular}}
\label{tab:performance}
\vspace{-5pt}
\end{table*}

In this section, we first compare the ranking relevance performance given different degrees of fairness requirements. Then we dive deep into our method to offer more insights into FARA's supremacy.
\subsubsection{Can FARA reach a better balance between fairness and effectiveness?}~
\label{sec:balancePost}
In Figure~\ref{fig:balance}, we compare ranking methods' effectiveness-fairness balance given different fairness requirements. To generate the balance curves in Figure~\ref{fig:balance}, we incrementally sample $\alpha$ from the minimum value to the maximum value within $\alpha$'s ranges indicated in Section~\ref{sec:baselines}. For each method, twenty $\alpha$  are sampled with the step size as $(\alpha_{max}-\alpha_{min})/20$. After sampling, we perform ranking simulation experiments for each $\alpha$ to get a (cNDCG, unfairness) pair. Then we connect different $\alpha$'s (cNDCG), unfairness) pair to form a curve for each method respectively in Figure~\ref{fig:balance}. All the curves start from the top right to the bottom left as  $\alpha$ increases, which means there exists a tradeoff between fairness and effectiveness (cNDCG). The reason behind this tradeoff is that requiring more fairness will bring more constraints on optimizing effectiveness.  Since TopK and RandomK  do not have trade-off parameters, both  of them only have one single pair of (cNDCG, unfairness), and their performances are shown as single points in Figure~\ref{fig:balance}. 

In Figure~\ref{fig:balance}, our methods FARA and FARA-Horiz. outperform all other fair methods since our methods reach the best cNDCG given the same unfairness tolerance. And FARA's supremacy is consistent in both post-processing and online settings. All fair ranking algorithms are effective fair ranking algorithms since they all show the tradeoff, i.e.,  higher cNDCG when increasing the unfairness tolerance. For unfair algorithms, TopK performs differently in post-processing and online settings.
In the post-processing setting, TopK reaches the highest cNDCG since relevance is known, and ranking relevance is the only consideration. However, in the online setting, TopK can not reach the highest cNDCG. We think the drop in cNDCG is that TopK naively trusts the relevance estimation without any exploration when optimizing ranking relevance. However, fair algorithms are shown to be robust to the online setting since they mostly can reach better cNDCG than Topk when increasing unfairness tolerance. We think the reason for the robustness is that fair algorithms usually rerank items for different sessions to optimize fairness, and such reranking brings explorations.
\vspace{-0.1cm}
\subsubsection{What is the fairness upper bound that FARA can reach?}\space
\label{sec:FairCapa}
In Table~\ref{tab:performance}, lower unfairness means higher fairness capacity and fairness upper bound, i.e., the maximum possible fairness one algorithm can reach. Fair effective ranking algorithms, including FairCo, LP, FARA-Horiz. and FARA, have similar fairness capacity and  outperform unfair ranking algorithms in terms of unfairness. The success of FARA-Horiz. and FARA  validates the proposed quadratic programming formulation can optimize fairness.  Similar cNDCG@5 and unfairness for those effective algorithms are expected according to Theorem~\ref{sec:fixedEffFair}.  For other fair ranking algorithms, ILP and MMF, and PLFair show inferior fairness capacity. As for the possible reason, ILP uses the integer linear programming method, which may not be effective in optimizing fairness. MMF actually follows a slightly different definition of fairness which require fairness at any cutoff should be fair, which is more strict than the definition we use in this paper. As for PLFair, PLFair tries to learn the ranking score that optimizes fairness based on the feature representation (the exact setting in original paper \cite{oosterhuis2021computationally}). However, the feature representation is initially designed for relevance which makes PLFair suboptimal for fairness optimization.    In Table~\ref{tab:performance}, ILP and LP are NA for MSLR10k and Istella-S due to time costs (refer to Table~\ref{tab:Timeperformance}).  Due to the page limit, we show the ranking performance of the online setting in Fig.~(\ref{fig:balance}), instead of in Table~\ref{tab:performance}.
\subsubsection{How is FARA's effectiveness at different cutoffs?}\space
In Table~\ref{tab:performance}, we show cNDCG at different cutoffs. Although FairCo, LP, FARA-Horiz. and FARA have similar fairness capacities, FARA significantly outperforms those fair algorithms for cNDCG$@1$ and cNDCG$@3$ on all three datasets. Compared to FARA-Horiz., shown in Table \ref{tab:performance},  FARA still significantly outperforms FARA-Horiz. at top ranks, which shows the necessity of vertical allocation. 


\subsubsection{How is FARA's time efficiency}? \space
\label{sec:effeciency}
\begin{table}[t]
\vspace{-5pt}
\centering
    \caption{The average time (seconds per 1k ranklists) cost with standard deviations in parentheses. Since ILP and LP are time-consuming on large datasets, the time costs  on MSLR-10k and Istella-S are estimated by only running 1k steps instead of the total simulation steps indicated in Sec.~\ref{sec:Simulation}.}
    \centering
  \resizebox{\columnwidth}{!}{
    \begin{tabular}{c c c c}\toprule
        \multirow{2}{*}{Algorithms} & \multicolumn{3}{c}{Datasets}\\ \cline{2-4}
       &{MSLR10k}&{Istella} &{MQ2008} \\ 
        \hline
        TopK&0.65$_{(0.14)}$
&0.50$_{(0.00)}$
&0.55$_{(0.10)}$

\\ 
        Randomk&0.63$_{(0.12)}$
&0.57$_{(0.04)}$
&0.59$_{(0.14)}$

\\ \cline{1-4}
        PLFair&2.24$_{(0.04)}$
&3.11$_{(0.07)}$
 &1.77$_{(0.04)}$

\\
        MMF&8.01$_{(0.39)}$
&6.57$_{(0.23)}$
&1.82$_{(0.28)}$

\\
ILP &1208.90$_{(85.80)}$ &1102.30$_{(75.20)}$&19.70$_{(1.29)}$
\\
        LP&$\ge$10 days&$\ge$10 days&2.09$_{(0.48)}$
 \\ 
        MCFair&0.724$_{(0.016)}$
&0.660$_{(0.025)}$
&0.567$_{(0.035)}$

\\         
        FairCo&0.73$_{(0.04)}$
&0.71$_{(0.03)}$
&0.70$_{(0.12)}$

\\ 
        FARA-Horiz.(Ours)&1.00$_{(0.17)}$
&0.86$_{(0.07)}$
&0.97$_{(0.22)}$

 \\
        FARA(Ours)&0.91$_{(0.07)}$
&0.91$_{(0.00)}$
&0.97$_{(0.30)}$

 \\
        \bottomrule
        \end{tabular} }
    \label{tab:Timeperformance}
    \vspace{-2pt}
\end{table}
Besides fairness and effectiveness optimization, we also empirically compare the time efficiency.
In Table~\ref{tab:Timeperformance}, ILP and LP are really time-consuming, especially on large datasets, MSLR10k and Istella-S. 
Compared to ILP and LP, FARA is more than $1000\times$ time efficient on MSLR10k and Istella-S, although all three of them are programming-based methods. There are two reasons behind FARA's time efficiency. The first one is that FARA has a much fewer number of decision variables since FARA only has $O(n)$ decision variables, while ILP and LP have $O(n^2)$ decision variables. The second one is that FARA does not need to solve quadratic programming for every time step. By solving quadratic programming once, we can get $\Delta T$ ranklists used for future $\Delta T$ sessions.
FARA can reach comparable time efficiency with non-programming-based algorithms like TopK, RandomK, and FairCo. Compared with those non-programming-based algorithms, the slightly additional time cost of FARA is acceptable given FARA' superior ranking performance in Tab.~\ref{tab:performance} and in Fig.~\ref{fig:balance}.


\subsubsection{How does $\Delta T$ influence FARA}?\space
\label{sec:FutureSess}
In Figure \ref{fig:perfAlongFuture}, we show the results of cNDCG and unfairness by varying the value of $\Delta T$. With greater $\Delta T$, we see a clear boost of \textit{cNDCG} for FARA, while such a boost does not happen for FARA-Horiz.  
We know that the proposed vertical allocation is the key reason to have better-ranking relevance when we get the optimal exposure planning $\Delta E^*$, and we theoretically analyze the reason in \S~\ref{sec:VerticalReason}.
Besides, as we increase $\Delta T$,  unfairness does not vary much, and its value stays close to the minimum unfairness we can achieve in Table~\ref{tab:performance}. We think the reason for the steady value of unfairness is that FARA already reaches the upper limit of  fairness when $\Delta T$ is small, and it is hard to improve when we increase $\Delta T$.
\begin{figure}[t]
    \centering
        \begin{subfigure}{0.23\textwidth}
    \includegraphics[scale=0.25]{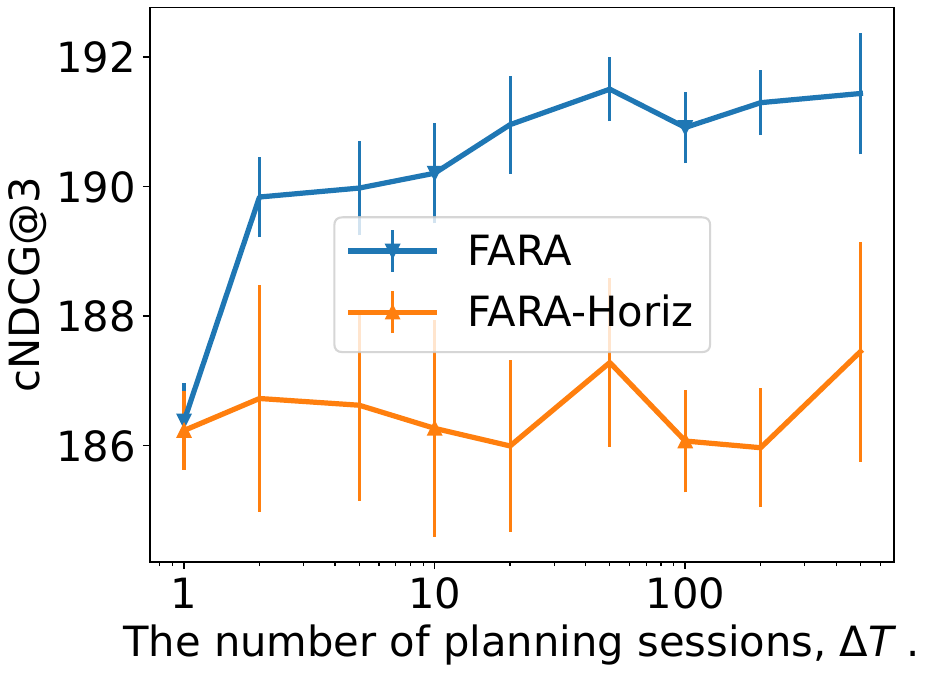}
         
    \end{subfigure} \hfill
        \begin{subfigure}{0.23\textwidth}
    \includegraphics[scale=0.25]{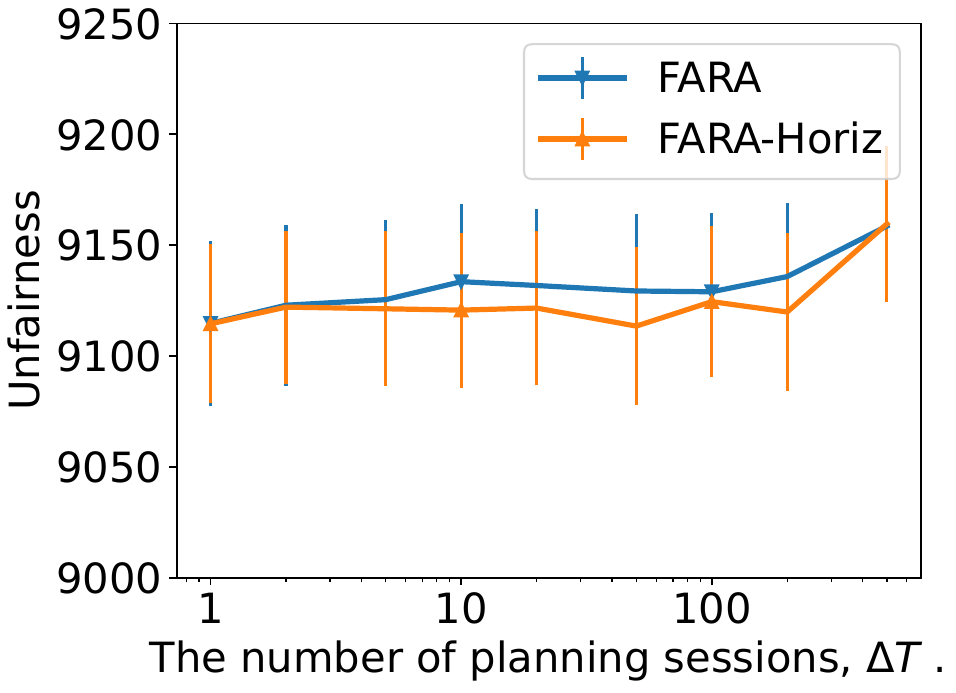}
    \end{subfigure}
    \caption{The numbers of planning session $\Delta T$'s influence on FARA in the post-processing
setting on MQ2008. $\alpha$ is set as 1. }
    \label{fig:perfAlongFuture}
\end{figure}
\begin{figure}[t]
    \centering
    \begin{subfigure}[]{0.23\textwidth}
    \includegraphics[scale=0.25]{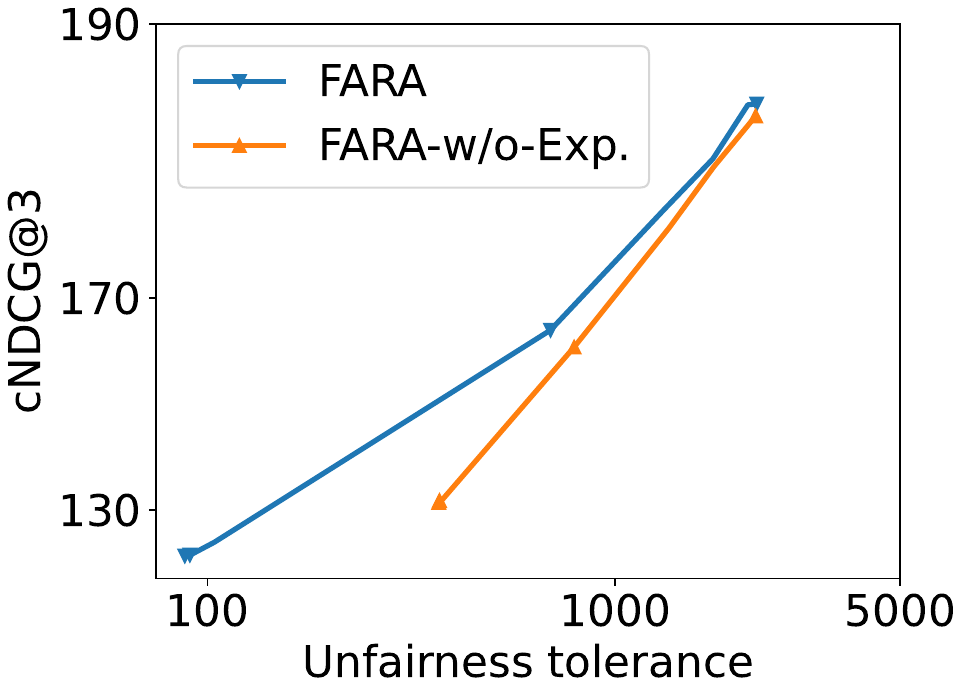}
    \caption{MSLR10k.}
    \end{subfigure}\hfill
    \begin{subfigure}[]{0.23\textwidth}
    \includegraphics[scale=0.25]{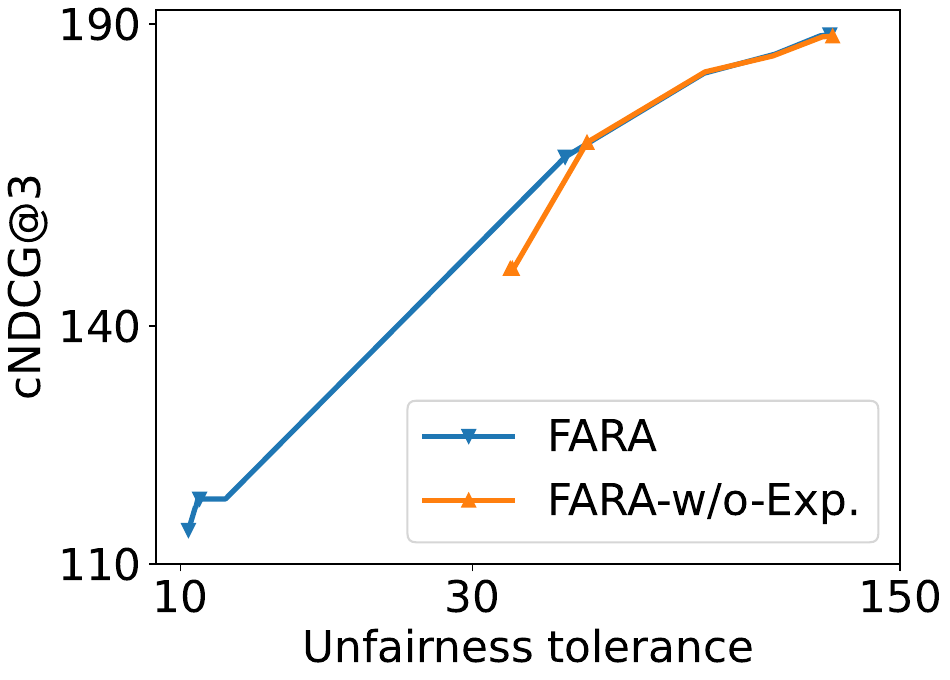}
    \caption{Istella-S}
    \end{subfigure}
    \caption{Ablation study of exploration in the online setting. The higher curves lie, the better their performances are.}
    \label{fig:abaltion}
\end{figure}
\subsubsection{How does exploration influence FARA in the online setting}?\space
To study how the exploration part (the slack variables $s$ in Eq.~\ref{eq:QPobjOn}) influences FARA, we did an ablation study for FARA  with or without exploration. Due to the page limit, we only  show the ablation results on the larger dataset, i.e., MSLR10k and Istella-S, in Figure~\ref{fig:abaltion}. The advantage of exploration is two-folded based on Figure~\ref{fig:abaltion}. Firstly, FARA lies higher than FARA-w/o-Exp. in the figure, which suggests exploration leads to a better effectiveness-fairness balance. Secondly, FARA has a smaller lower bound of unfairness tolerance,  which implies exploration enables FARA to have a higher fairness capacity and can meet a more strict fairness requirement.

\section*{Acknowledgements}
\vspace{-5pt}
This work was supported in part by NSF CCF-2115677 and in part by the School of Computing, University of Utah. Any opinions, findings, conclusions, or recommendations expressed in this material are those of the authors and do not necessarily reflect those of the sponsor.
\clearpage
\newpage
\bibliographystyle{ACM-Reference-Format}
\balance
\bibliography{mybib}
\end{document}